\documentclass[
    aps,
    reprint,
    superscriptaddress,
    longbibliography,
    ]{revtex4-2}

\usepackage{float}
\makeatletter
\let\newfloat\newfloat@ltx
\makeatother

\usepackage{algorithm}
\usepackage{algorithmicx}
\usepackage{algpseudocode}
\renewcommand{\algorithmicrequire}{\textbf{Input:}}
\renewcommand{\algorithmicensure}{\textbf{Output:}}

\usepackage{graphicx}
\usepackage{amsmath,amssymb,bm}
\usepackage{amsthm}
\theoremstyle{definition}
\newtheorem{definition}{Definition}
\newtheorem{theorem}{Theorem}

\usepackage{braket}
\usepackage{multirow}
\usepackage{mathtools}
\mathtoolsset{showonlyrefs}
\usepackage{hyperref}
\usepackage{xcolor}

\begin{document}

\title{Even More Efficient Soft-Output Decoding with Extra-Cluster Growth and Early Stopping}

\author{Kaito Kishi}
\thanks{These authors contributed equally to this work.\\
Email: kishi.kaito@fujitsu.com}
\affiliation{Quantum Laboratory, Fujitsu Research, Fujitsu Limited, 4-1-1 Kawasaki, Kanagawa 211-8588, Japan}
\affiliation{
Fujitsu Quantum Computing Joint Research Division,
Center for Quantum Information and Quantum Biology, Osaka University, 1-2 Machikaneyama, Toyonaka, Osaka, 565-8531, Japan
}

\author{Riki Toshio}
\thanks{These authors contributed equally to this work.\\
Email: kishi.kaito@fujitsu.com}
\affiliation{Quantum Laboratory, Fujitsu Research, Fujitsu Limited, 4-1-1 Kawasaki, Kanagawa 211-8588, Japan}
\affiliation{
Fujitsu Quantum Computing Joint Research Division,
Center for Quantum Information and Quantum Biology, Osaka University, 1-2 Machikaneyama, Toyonaka, Osaka, 565-8531, Japan
}

\author{Jun Fujisaki}
\affiliation{Quantum Laboratory, Fujitsu Research, Fujitsu Limited, 4-1-1 Kawasaki, Kanagawa 211-8588, Japan}
\affiliation{
Fujitsu Quantum Computing Joint Research Division,
Center for Quantum Information and Quantum Biology, Osaka University, 1-2 Machikaneyama, Toyonaka, Osaka, 565-8531, Japan
}

\author{Hirotaka Oshima}
\affiliation{Quantum Laboratory, Fujitsu Research, Fujitsu Limited, 4-1-1 Kawasaki, Kanagawa 211-8588, Japan}
\affiliation{
Fujitsu Quantum Computing Joint Research Division,
Center for Quantum Information and Quantum Biology, Osaka University, 1-2 Machikaneyama, Toyonaka, Osaka, 565-8531, Japan
}

\author{Shintaro Sato}
\affiliation{Quantum Laboratory, Fujitsu Research, Fujitsu Limited, 4-1-1 Kawasaki, Kanagawa 211-8588, Japan}
\affiliation{
Fujitsu Quantum Computing Joint Research Division,
Center for Quantum Information and Quantum Biology, Osaka University, 1-2 Machikaneyama, Toyonaka, Osaka, 565-8531, Japan
}

\author{Keisuke Fujii}
\affiliation{
Fujitsu Quantum Computing Joint Research Division,
Center for Quantum Information and Quantum Biology, Osaka University, 1-2 Machikaneyama, Toyonaka, Osaka, 565-8531, Japan
}

\affiliation{
Graduate School of Engineering Science, Osaka University,
1-3 Machikaneyama, Toyonaka, Osaka, 560-8531, Japan
}
\affiliation{
Center for Quantum Information and Quantum Biology, Osaka University, 560-0043, Japan
}
\affiliation{
RIKEN Center for Quantum Computing (RQC), Wako Saitama 351-0198, Japan
}

\date{\today}

\begin{abstract}
    In fault-tolerant quantum computing, soft outputs from real-time decoders play a crucial role in improving decoding accuracy, post-selecting magic states, and accelerating lattice surgery.
    A recent paper by Meister \textit{et al}. [arXiv:2405.07433 (2024)] proposed an efficient method to evaluate soft outputs for cluster-based decoders, including the Union-Find (UF) decoder.
    However, in parallel computing environments, its computational complexity is comparable to or even surpasses that of the UF decoder itself, resulting in a substantial overhead.
    Furthermore, this method requires global information about the decoding graph, making it poorly suited for existing hardware implementations of the UF decoder on Field-Programmable Gate Arrays (FPGAs).
    In this paper, to alleviate these issues, we develop more efficient methods for evaluating high-quality soft outputs in cluster-based decoders by introducing several \textit{early-stopping} techniques.
    Our central idea is that the precise value of a large soft output is often unnecessary in practice.
    Based on this insight, we introduce two types of novel soft-outputs: the \textit{bounded cluster gap} and the \textit{extra-cluster gap}.
    The former reduces the computational complexity of Meister's method by terminating the calculation at an early stage.
    Our numerical simulations show that this method achieves improved scaling with code distance $d$ compared to the original proposal.
    The latter, the extra-cluster gap, quantifies decoder reliability by performing a small, additional growth of the clusters obtained by the decoder.
    This approach offers the significant advantage of enabling soft-output computation without modifying the existing architecture of FPGA-implemented UF decoders.
    These techniques offer lower computational complexity and higher hardware compatibility, laying a crucial foundation for future real-time decoders with soft outputs.
\end{abstract}

\maketitle

\section{Introduction}

Quantum computers hold great promise for a wide range of applications, including quantum chemistry~\cite{cao2019quantumchemistry}, cryptography~\cite{shor1994algo, gidney2025howtofactor}, and machine learning~\cite{biamonte2017qml}.
Realizing these applications requires fault-tolerant quantum computers (FTQCs), which rely on the quantum error correction (QEC) schemes.
Among the various QEC codes, the surface code~\cite{kitaev2003faulttolerant, bravyi1998quantumcodesboundary} is a particularly promising candidate due to its high error threshold and its implementation requiring only nearest-neighbor interactions.

In QEC, a decoder estimates the errors that have occurred on the logical qubits.
The accuracy of the decoder directly impacts the logical error probability.
Therefore, if accuracy were the sole concern, a maximum-likelihood (ML) decoder would be the optimal choice.
This decoder exactly identifies the most probable logical error, while causing an exponential time overhead in general.
However, real-time decoding is required repeatedly throughout a quantum computation, for instance, to handle non-Clifford gates.
If the decoding time exceeds the syndrome generation time, the \textit{backlog problem} occurs, leading to an exponential increase in total computation time~\cite{terhal2015backlog}.
Consequently, a practical decoder must achieve a balance between high accuracy and high speed.

Recently, quantifying the reliability of decoder's estimates has emerged as a promising solution to the above issue~\cite{toshio2025decoderswitching}.
This reliability metric is commonly referred to as a decoder's \textit{soft output}, which provides \textit{soft information} about the confidence in a decoding result.
Such information plays a pivotal role in the \textit{decoder-switching} framework proposed in Ref.~\cite{toshio2025decoderswitching}, which combines paired complementary decoders adaptively to realize high-speed, high-accuracy real-time decoding.
More specifically, in this framework, a fast low-accuracy soft-output decoder (\textit{weak decoder}) is used for usual rounds, while a slower high-accuracy decoder (\textit{strong decoder}) is invoked only when the soft output of the weak decoder indicates low confidence.
This enables us to achieve the accuracy of the strong decoder at the high decoding speed of the weak one.
The libra decoder~\cite{jones2024improvedaccuracy} employs a similar concept, running an ensemble of decoders only for low-confidence results to boost accuracy with minimal overhead.

Beyond optimizing the decoding process itself, soft information has a broad range of other applications.
For instance, it can be used for post-selection to enhance the effective code distance by discarding outcomes deemed unreliable~\cite{smith2024mitigatingerrorsinlq,meister2024efficientsoftoutput,dinca2025errormitigation,zhou2025errormitigation,sunami2025entanglementboosting}.
A similar technique is applied to filter states during magic state distillation~\cite{bombin2024faulttolerantmsp} and cultivation~\cite{gidney2024msc,hirano2025efficientmsc}.
In concatenated codes, the soft output from an inner code can be passed as soft information to an outer code to improve overall performance~\cite{gidney2025yokedsurfacecodes,meister2024efficientsoftoutput}.
Furthermore, soft information has recently been proposed to dynamically reduce the runtime of lattice surgery operations~\cite{akahoshi2025runtimereduction}.

While soft output can be naturally obtained from methods like tensor network decoders, these approaches are computationally expensive, often requiring exponential time.
More recently, the concept of the \textit{complementary gap} (also known as the logical gap) has been introduced, enabling the efficient calculation of soft output~\cite{gidney2025yokedsurfacecodes,bombin2024faulttolerantmsp,smith2024mitigatingerrorsinlq}.
Subsequently, Meister \textit{et al}. proposed an efficient method for computing soft output specifically for cluster-based decoders~\cite{meister2024efficientsoftoutput}.
In this paper, we refer to this method as the \textit{cluster gap} (also known as the swim distance~\cite{dinca2025errormitigation}).
With recent progress in formulating soft output for qLDPC codes~\cite{lee2025efficientpostselection}, its importance in FTQC studies continues to grow.

To realize practical real-time decoders, minimizing the computational overhead of soft-output calculation is critical. However, existing methods such as the complementary and cluster gaps introduce non-negligible time overhead, often comparable to the decoding algorithm itself~\cite{jones2024improvedaccuracy,meister2024efficientsoftoutput}. This computational cost is particularly problematic for Union-Find (UF) decoders, which are typically designed for parallel implementation on Field-Programmable Gate Arrays (FPGAs)~\cite{liyanage2024heliosv2,liyanage2025heliosv3,valentino2025quekuf,heer2023noveluf,heer2023achievingscalable,barber2025collisionclustering}. In such hardware implementations, the soft-output overhead becomes relatively significant since the overhead of the UF decoder itself can be reduced to sublinear in the code distance $d$~\cite{liyanage2024heliosv2}, falling below that of calculating the cluster gap~\cite{meister2024efficientsoftoutput}.
This challenge is further exacerbated in QEC codes with multiple logical degrees of freedom, as soft output must be computed for every combination of logical operators.

In this work, we address the computational bottleneck of soft-output calculation for cluster-based decoders in real-time and parallel-computing environments.
Our main contribution is the introduction of two complementary concepts for fast and reliable confidence estimation:
{\it extra-cluster growth} and {\it early stopping}.
These ideas work together to enable accurate soft-output evaluation, while substantially reducing the computational overhead compared to existing methods.

The first concept, extra-cluster growth, introduces a new paradigm for confidence estimation in cluster-based decoders.
Instead of computing the decoding result and its confidence through separate post-processing steps, we estimate the decoder confidence by performing a controlled, additional growth of clusters after the decoding has completed.
This allows decoding and confidence estimation to be carried out within a single cluster-growth framework.
As a result, the soft-output calculation can directly reuse the cluster growth module of the decoder, eliminating the need for a separate shortest-path computation and making the method highly compatible with FPGA-based implementations of UF decoders.
The second concept, early stopping, is a general strategy applicable to confidence estimation.
It is motivated by the observation that many practical applications---such as decoder switching and post-selection---do not require the precise value of a large soft output.
Instead, it is sufficient to determine whether the confidence is below a predefined threshold.
By terminating the calculation as soon as this condition is resolved, early stopping significantly reduces computational cost without sacrificing relevant information.

To first isolate the effect of early stopping, we apply it to the existing cluster-gap calculation based on Dijkstra’s algorithm~\cite{dijkstra1959anote}.
This leads to the {\it bounded cluster gap}, which terminates the shortest-path search once the distance is guaranteed to exceed the threshold.
Numerical simulations show a substantial reduction in computational cost in the low-error regime.
For example, at a physical error rate of $p=0.10\%$, the number of nodes visited during the calculation scales approximately as $O(d^{2.31})$, compared to $O(d^{2.88})$ for the original cluster gap.
At lower error rates, the reduction is even more pronounced, reaching nearly two orders of magnitude at $p=0.05\%$.
These results demonstrate that early stopping alone can significantly mitigate the computational overhead of soft-output calculation.

However, the bounded cluster gap still relies on a Dijkstra-based search, which constitutes a process separate from the decoding itself.
To address this issue, we also introduce the \textit{extra-cluster gap}, an alternative estimator that quantifies decoder confidence by performing a small, additional growth of the clusters.
A key advantage of this method is its high compatibility with existing hardware architectures of UF decoders, as it reuses the core cluster growth module.
Besides hardware compatibility, we theoretically prove that, despite its simplicity, this approach is guaranteed to identify every instance where the original cluster gap is below a predefined threshold.
This feature ensures that no low-confidence results are missed, making it a reliable tool for applications like decoder switching. Our numerical analyses confirm the practical efficacy of extra-cluster gaps in the decoder-switching framework; for a distance-25 surface code at a physical error rate of $p = 0.10\%$, the extra-cluster gap predicts a switching rate as low as approximately $4 \times 10^{-10}$, which is small enough to prevent the backlog problem in decoder switching.
Furthermore, the benefits of the extra-cluster gap become even more pronounced in complex QEC architectures with multiple logical qubits.
For a QEC system with $M$ non-equivalent boundaries, calculating the soft output for all pairs requires $O(M^2)$ computations for the complementary gap.
In contrast, our extra-cluster gap requires only a single computation, drastically reducing the overall complexity of soft-output calculations.

In conclusion, the bounded cluster gap serves as a reference that quantifies the benefit of early stopping in isolation, while the extra-cluster gap provides a fully integrated, hardware-friendly solution that combines early stopping with extra-cluster growth.
This combination enables fast, scalable soft-output calculation suitable for real-time decoding, decoder switching, and QEC architectures with multiple logical boundaries.

The remainder of this paper is organized as follows.
Section~\ref{sec:background} reviews the fundamentals of QEC and existing methods for soft-output calculation.
Section~\ref{sec:early-stopping} provides the technical details of our proposed methods, the bounded cluster gap and the extra-cluster gap.
In Section~\ref{sec:numerical} we present numerical simulations to evaluate the performance of our proposals.
In Section~\ref{sec:applications}, we discuss the practical applications of our findings, such as for decoder switching and in scenarios with multiple logical boundaries.
Finally, Section~\ref{sec:conclusion} summarizes our findings and outlines future prospects.

\begin{figure*}[htb]
    \includegraphics[scale=0.5]{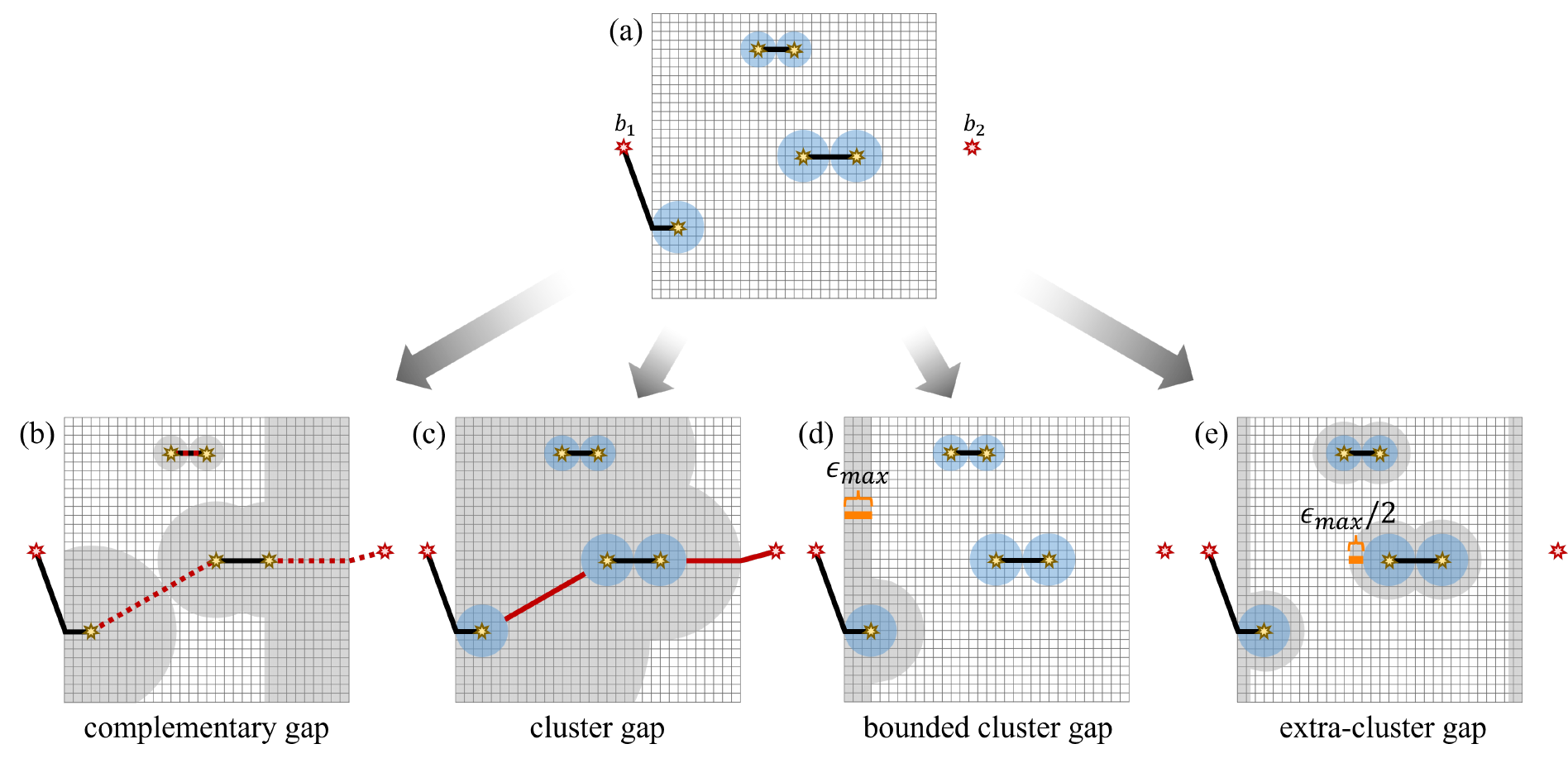}
    \caption{
        Schematic illustrations of different methods for calculating soft outputs.
        ({\bf a}) A cluster-based decoder forms multiple clusters (blue circles) around each error syndrome (yellow star) to determine an initial error-correction path (black line) for a given decoding problem.
        ({\bf b}) The \textit{complementary gap}~\cite{gidney2025yokedsurfacecodes} is the weight difference between the initial correction and the optimal correction for the complementary logical class (red dotted line).
        The gray area indicates the search space explored during the complementary decoding.
        ({\bf c}) The \textit{cluster gap}~\cite{meister2024efficientsoftoutput} is the shortest path distance between boundaries $b_1$ and $b_2$ (red line), calculated using Dijkstra's algorithm after all intra-cluster edge weights are set to zero.
        The gray area represents the region explored by the algorithm.
        ({\bf d}) The \textit{bounded cluster gap}, one of the methods we introduce, modifies the cluster gap by terminating the Dijkstra's search early.
        This process stops once the path distance is guaranteed to exceed a predefined threshold $\epsilon_\mathrm{max}$, which significantly reduces the search area (gray).
        ({\bf e}) The \textit{extra-cluster gap}, another method we propose, is determined by an additional growth of all clusters.
        The calculation is terminated if no single grown cluster connects boundaries $b_1$ and $b_2$ within a growth limit of $\epsilon_\mathrm{max}/2$.
        The gray area shows the region covered by this additional growth.
    }
    \label{fig:method-diagram}
\end{figure*}

\section{Background} \label{sec:background}
\subsection{Decoding Graph}

A Calderbank-Shor-Steane (CSS) code is defined by its check matrices $H_X$ and $H_Z$, whose rows represent the $X$- and $Z$-stabilizer generators, respectively.
In what follows, we neglect the correlation between $X$- and $Z$-errors and, for simplicity, focus on decoding the $Z$-errors.
Among the various types of CSS codes, this work will focus primarily on the surface code.

For the $X$-stabilizers of the surface code, we can construct a decoding graph where the nodes represent the detectors (i.e., the rows of $H_X$) and the edges represent physical errors that cause detector to flip.
Due to the geometric locality of the surface code, an edge connects either one or two detectors.
An edge incident to only one detector is called a half-edge and is treated as connecting to a virtual boundary node, denoted as $b_1$.

To calculate metrics such as the complementary gap or the cluster gap, this decoding graph is slightly modified~\cite{gidney2025yokedsurfacecodes,meister2024efficientsoftoutput}.
Specifically, the nodes connected to the original boundary node $b_1$ are partitioned.
The nodes corresponding to one side of the graph are rewired to a new, separate boundary node, denoted as $b_2$.
These two nodes are collectively referred to as the \textit{inequivalent boundaries}, denoted as $B=\{b_1, b_2\}$.
In this modified graph, a path connecting $b_1$ and $b_2$ constitutes a logical operator.
We denote this graph as $G=(V,E)$, where $V$ is the set of detectors plus $B$, and $E$ is the set of edges representing possible physical errors.

Assuming a circuit-level noise model, the maximum degree of any node in the decoding graph is 12. Each edge $e\in E$ is assigned a weight $w_e=\log{((1-p_e)/p_e)}$, where $p_e$ is the corresponding error probability.
When a distance-$d$ surface code is idled for $d$ rounds, the graph $G$ contains $O(d^3)$ detectors.
Meanwhile, the number of detectors adjacent to each boundary, $b_1$ and $b_2$, scales as $O(d^2)$.
In the following, we will assume these scalings for the number of detectors on $G$.

\subsection{Cluster-based Decoder}

Physical errors can flip the state of detectors, and the locations of these flips are referred to as detection events~\cite{gidney2021stim}.
A decoder's objective is to estimate the most likely logical error from a given set of detection events.
The decoding is deemed successful if the product of the true error and the applied correction is a trivial logical operator; otherwise, a logical error occurs.

Prominent examples of decoders include the minimum-weight perfect matching (MWPM) decoder~\cite{higgott2025sparseblossom,wu2023fusionblossom}, which is a most-likely error (MLE) decoder for codes with a matchable error graph, and its approximation, the UF decoder~\cite{delfosse2020lineartimemaximumlikelihood,wu2022interpretationunionfind}.
Both are cluster-based decoders that operate by growing clusters from detection events on the decoding graph.
In this approach, a cluster is formally defined as the union of balls of a certain radius centered at each detection event~\cite{meister2024efficientsoftoutput}.
This process continues until every detection event is paired within a cluster or matched to a boundary.
The growth mechanisms differ between the two: the MWPM decoder uses alternating trees and blossoms, whereas the standard UF decoder expands all active clusters uniformly.
Due to its algorithmic simplicity and amenability to parallelization, the UF decoder has been implemented on dedicated hardware like FPGAs, achieving much higher throughput than CPU-based implementations~\cite{liyanage2024heliosv2,valentino2025quekuf,heer2023achievingscalable,heer2023noveluf,ziad2024localclusteringdecoder,barber2025collisionclustering}.
For instance, Ref.~\cite{liyanage2024heliosv2} implemented a UF decoder on a Xilinx VCU129 FPGA, demonstrating that it can solve a $d=51$ decoding problem with a phenomenological noise model in 544~ns per round.
Notably, this hardware implementation achieves sublinear average-case time complexity, making it significantly more scalable than sequential software versions.

\subsection{Soft-Output Calculation}
\label{sec:soft-output}

A soft output is a metric that quantifies the reliability of a decoder's output.
A well-known example is the complementary gap, which can be efficiently computed by MLE decoders~\cite{gidney2025yokedsurfacecodes,bombin2024faulttolerantmsp,smith2024mitigatingerrorsinlq}.
It is defined as the weight difference between the two minimum-weight perfect matchings corresponding to different logical outcomes (see Figure~\ref{fig:method-diagram}~(b)).
The primary drawback of this method is the high computational cost of performing a second decoding to find the most-likely matching for the complementary logical class~\cite{jones2024improvedaccuracy}.
This second step is particularly time-consuming at low physical error rates, as it requires significant cluster growth to find a complementary matching.

To address this high cost, an efficient alternative known as the cluster gap has been proposed (termed following Ref.~\cite{toshio2025decoderswitching})~\cite{meister2024efficientsoftoutput}.
The calculation involves several steps, as shown in Figure~\ref{fig:method-diagram}~(c).
First, an initial decoding is performed using a cluster-based decoder.
The resulting set of final clusters is then used to define a new contracted graph, $G'$, where each cluster from the original graph $G$ is condensed into a single node.
Mathematically, $G'$ is the quotient graph of $G$ with respect to the partition defined by the clusters (see Definition~9 of Ref.~\cite{meister2024efficientsoftoutput}).
This contraction is equivalent to setting the weights of all edges within the clusters to zero.
Finally, the soft-output value is determined by calculating the shortest distance between inequivalent boundaries on $G'$ using Dijkstra's algorithm.

While this approach avoids the costly second decoding, the use of Dijkstra's algorithm still incurs a time complexity of $O(d^3 \log d)$~\cite{dijkstra1959anote}.
This is slightly worse than the complexity of the UF decoder, which is nearly linear at $O(d^3 \alpha(d^3))$, where $\alpha$ is the inverse Ackermann function~\cite{tarjan1975efficiency,delfosse2020lineartimemaximumlikelihood}.
This performance gap becomes a more significant bottleneck in parallel computing environments.
As previously mentioned, the complexity of a parallelized UF decoder scales sublinearly with $d$, making it substantially more efficient than the cluster gap calculation.

\section{Early Stopping and Extra-Cluster Growth} \label{sec:early-stopping}

In this section, to accelerate soft-output calculation, we introduce two complementary approaches: early stopping and extra-cluster growth.
These strategies leverage a predefined soft-output threshold, $\epsilon_\mathrm{max}$, whose value is determined by the criteria for post-selection~\cite{smith2024mitigatingerrorsinlq,bombin2024faulttolerantmsp,gidney2024msc,hirano2025efficientmsc,dinca2025errormitigation,zhou2025errormitigation,sunami2025entanglementboosting} or some switching methods~\cite{jones2024improvedaccuracy,toshio2025decoderswitching}.

\subsection{Bounded Cluster Gap}\label{sec:bounded-cluster-gap}

To first isolate and quantify the standalone benefit of early stopping, we apply this strategy to the existing cluster gap calculation.
This procedure gives rise to what we term the bounded cluster gap, a method that reduces the time complexity of the cluster gap calculation.
This approach leverages the operational principle of Dijkstra's algorithm, which systematically explores graph nodes in increasing order of distance from a source using a priority queue.
Consequently, the search can be terminated as soon as the distance of the node extracted from the priority queue exceeds the threshold $\epsilon_\mathrm{max}$.
This modified version of the algorithm is known as bounded Dijkstra's algorithm, and its performance has been previously analyzed in detail~\cite{bemten2019boundeddijkstra}.
In this work, we analyze the performance of the bounded cluster gap when applied to the decoding graph of surface codes.

Figure~\ref{fig:method-diagram}~(c) and (d) illustrate the search spaces for the original cluster gap and the bounded cluster gap, respectively.
The gray area in Figure~\ref{fig:method-diagram}~(d) shows that the bounded cluster gap confines the search space to a narrower region than the original cluster gap in Figure~\ref{fig:method-diagram}~(c).
If we assume $\epsilon_\mathrm{max}$ is a constant independent of the code distance $d$, the search is limited to a radius of approximately $\epsilon_\mathrm{max}$ from the boundary node $b_1$.
At a low physical error probability $p$, large error clusters are unlikely to form near $b_1$.
Therefore, the search space is confined to the vicinity of the boundary, containing a number of nodes on the order of $O(d^2)$.
Since Dijkstra's algorithm has a time complexity of $O(N\log N)$ for a graph with $N$ nodes and $O(N)$ edges~\cite{dijkstra1959anote}, the average time complexity for the bounded cluster gap in this low-error regime is
\begin{align}
O(d^2\log(d^2))=O(d^2\log d), \label{eq:bounded-cluster-gap-time-complexity}
\end{align}
where we used the fact that $N=O(d^3)$ in typical decoding problems for distance-$d$ surface codes.

Conversely, at a high physical error probability $p$, the likelihood of a large cluster forming adjacent to $b_1$ increases.
Since edge weights are zero within such a cluster, the search can traverse a large area while remaining within the distance limit $\epsilon_\mathrm{max}$.
In the worst-case scenario, where the cluster spans the entire graph $G$, the time complexity reverts to $O(d^3\log d)$, matching that of the original cluster gap.
We will later present numerical experiments to demonstrate the relationship between $p$ and the effective search space.
While a more efficient shortest-path algorithm was recently discovered~\cite{duan2025fasterdijkstra}, its improvement changes the complexity's logarithmic factor from $\log(\cdot)$ to $\log^{2/3}(\cdot)$, which does not significantly alter our conclusions.

Our discussion thus far has focused on sequential computation.
For parallel computation, alternative shortest-path algorithms exist, such as the $\Delta$-stepping algorithm~\cite{meyer2003deltastepping} and its derivatives~\cite{dong2021efficientstepping,vedadi2025hybstepping}.
The time complexity of the $\Delta$-stepping algorithm is $O(L\log N)$ for a graph with $N$ nodes, $O(N)$ edges, constant maximum node degree, and a shortest-path length of $L$~\cite{meyer2003deltastepping}.
In the low-error regime, the path length $L$ is typically small and bounded by $\epsilon_\mathrm{max}$, reducing the average time complexity to $O(\log d)$.
However, similar to the sequential case, for large $p$, $L$ can be on the order of $d$, leading to a time complexity of $O(d\log d)$.

\subsection{Extra-Cluster Gap}\label{sec:extra-cluster-gap}

In this section, we propose an alternative type of soft output called the extra-cluster gap.
This is designed for efficient soft-output calculation on dedicated hardware, such as FPGAs, based on the existing cluster-growth modules of cluster-based decoders.

The basic idea behind the extra-cluster gap stems from reinterpreting the cluster gap within the framework of ``extra-cluster growth."
As explained in Section~\ref{sec:soft-output}, the cluster gap quantifies the decoder's confidence by measuring the shortest distance between the non-equivalent boundaries on $G$ after removing the weights on clusters.
Our key insight is that an equivalent quantity can be reconstructed by introducing the extra-cluster growth process as follows: 
First, cluster-based decoding is performed to solve a specific decoding task, thereby forming corresponding clusters on the decoding graph. Next, the resulting clusters are grown additionally until the non-equivalent boundaries become connected via these clusters. Finally, the amount of growth required for this connection is quantified, which yields a quantity equivalent to the cluster gap. In fact, we theoretically and numerically confirm the equivalence or relationship between these approaches in the subsequent discussions.
Importantly, this new insight offers an opportunity to design soft-output calculations more flexibly. In this work, by setting a cutoff for additional growth, we formulate the extra-cluster gap as a novel soft output that efficiently approximates the cluster gap.

In what follows, we present two variants of the extra-cluster gap.
The first, simplified one relies solely on the additional growth procedure and is referred to as the extra-cluster gap without cluster graph (w/o CG).
The second, more precise one constructs a \textit{cluster graph} from the inter-cluster distances to yield a result identical to the original cluster gap, which we call the extra-cluster gap with cluster graph (w/ CG).

\subsubsection{Extra-Cluster Gap without Cluster Graph (w/o CG)}
First, we describe the simpler approach, the extra-cluster gap w/o CG, which is detailed in Algorithm~\ref{alg:extra-cluster-gap-wo-cluster-graph}.
This approach additionally grows all clusters by a radius below $\epsilon_\mathrm{max}/2$.
During this process, the decoder checks if a single cluster that connects the boundaries $b_1$ and $b_2$ is formed.
If such a connection occurs, the algorithm returns the minimum growth amount required for the connection as the soft-output value.
If no connection is formed within the growth limit, it signifies that no soft-output value was found in that range.

\begin{algorithm}[t]
    \caption{Extra-cluster gap without the cluster graph}
    \label{alg:extra-cluster-gap-wo-cluster-graph}
    \begin{algorithmic}[1]
        \Require A graph $G$ with clusters and boundaries $b_1$, $b_2$.
        \Ensure A soft-output value if boundaries are connected by additional growth up to $\epsilon_\mathrm{max}$, otherwise null.
        \While {$\epsilon \leq \epsilon_\mathrm{max}$}
            \State Increase the radii of all clusters and boundary nodes by $\delta\epsilon/2$ ($\delta\epsilon$: arbitrary).
            \State Merge colliding clusters into a single cluster using the Union operation.
            \State $\epsilon\gets \epsilon + \delta\epsilon$
            \If {a single cluster connects boundaries $b_1$ and $b_2$}
                \State \Return $\epsilon$
            \EndIf
        \EndWhile
        \State \Return null
    \end{algorithmic}
\end{algorithm}

To analyze this approach theoretically, here we define the cluster gap and the extra-cluster gap w/o CG formally as follows:

\begin{definition}[Cluster Gap: $g_\mathrm{c}$] \label{def:g_c}
    Let $P_\mathrm{c}$ be the shortest path connecting the boundary nodes $b_1$ and $b_2$ in $G'$.
    The cluster gap, $g_\mathrm{c}$, is defined as the total weight of this path.

\end{definition}

\begin{definition}[Extra-Cluster Gap w/o CG: $g_\mathrm{ec}$] \label{def:g_ec}
    The extra-cluster gap without a cluster graph, $g_\mathrm{ec}$, is defined based on a search over a growth parameter $\epsilon$.
    For a given $\epsilon$, let $G'_\epsilon$ be the subgraph of $G'$ that includes only edges with weights less than or equal to $\epsilon$ from each cluster or boundary node.
    We define $g_\mathrm{ec}$ as the minimum value of $\epsilon\in[0,\epsilon_\mathrm{max}]$ for which a path exists between the boundary nodes $b_1$ and $b_2$ in $G'_\epsilon$.
    If no such path is found for any $\epsilon\leq\epsilon_\mathrm{max}$, $g_\mathrm{ec}$ is undefined.

\end{definition}
Here we note that Definition~\ref{def:g_c} is identical to $\phi(\mathcal{C})$ in Definition~9 of Ref.~\cite{meister2024efficientsoftoutput}.

To facilitate the proof, we introduce an additional definition related to the path $P_\mathrm{c}$.

\begin{definition}[Maximum Inter-Cluster Edge Weight on $P_\mathrm{c}$: $w_\mathrm{max}(P_\mathrm{c})$]
    The path $P_\mathrm{c}$ connects the boundaries $b_1$ and $b_2$ by traversing a sequence of zero or more clusters.
    We define $w_\mathrm{max}(P_\mathrm{c})$ as the maximum of the total weights of all edges connecting any two consecutive elements (clusters or boundaries) along the path $P_\mathrm{c}$.
\end{definition}

The relationship between $g_\mathrm{c}$ and $g_\mathrm{ec}$ is summarized by the following theorems.

\begin{theorem} \label{theorem:extra-cluster-gap-wo-cg-general}
    For any threshold $\epsilon_\mathrm{max}\geq 0$, one of two conditions must hold:
    \begin{itemize}
        \item $g_\mathrm{ec}$ is defined, and it satisfies the inequality $g_\mathrm{ec}\leq g_\mathrm{c}$.
        \item $g_\mathrm{ec}$ is undefined.
    \end{itemize}

\end{theorem}
\begin{proof}
    Consider the shortest path $P_\mathrm{c}$, which has a total weight of $g_\mathrm{c}$. 
    If we set the growth parameter $\epsilon$ to be $w_\mathrm{max}(P_\mathrm{c})$, all edges of the path $P_\mathrm{c}$ are included in the subgraph $G'_\epsilon$.
    This ensures that a path connecting $b_1$ and $b_2$ exists in $G'_\epsilon$ for $\epsilon = w_\mathrm{max}(P_\mathrm{c})$.
    Since $g_\mathrm{ec}$ is the minimum such $\epsilon$ for which a path exists, we have $g_\mathrm{ec}\leq w_\mathrm{max}(P_\mathrm{c})$.

    Furthermore, the maximum weight of a single connection between consecutive elements on a path, $w_\mathrm{max}(P_\mathrm{c})$, cannot exceed the total weight of the entire path, $g_\mathrm{c}$.
    The total weight $g_\mathrm{c}$ is the sum of all such connection weights.
    This gives the inequality $w_\mathrm{max}(P_\mathrm{c}) \leq g_\mathrm{c}$.

    Combining these results, we find that $g_\mathrm{ec} \leq g_\mathrm{c}$ if $g_\mathrm{ec}$ is defined.
    If no path satisfies the condition for any $\epsilon \leq \epsilon_\mathrm{max}$, then $g_\mathrm{ec}$ is undefined.

\end{proof}

\begin{theorem} \label{theorem:extra-cluster-gap-wo-cg-special}
    If the cluster gap $g_{\mathrm{c}}$ is less than or equal to the threshold $\epsilon_\text{max}$, then $g_{\mathrm{ec}}$ is guaranteed to be defined and satisfies $g_{\mathrm{ec}} \leq g_{\mathrm{c}}$.

\end{theorem}
\begin{proof}
    From Theorem~\ref{theorem:extra-cluster-gap-wo-cg-general}, we know that $g_{\mathrm{ec}} \leq g_{\mathrm{c}}$ whenever $g_{\mathrm{ec}}$ is defined.
    We therefore only need to show that the condition $g_{\mathrm{c}}\leq \epsilon_\text{max}$ guarantees that $g_{\mathrm{ec}}$ is defined.

    As shown in the proof of Theorem~\ref{theorem:extra-cluster-gap-wo-cg-general}, $w_\mathrm{max}(P_{\mathrm{c}})$ is less than or equal to $g_\mathrm{c}$.
    The condition $g_{\mathrm{c}}\leq \epsilon_\mathrm{max}$ therefore implies $w_\mathrm{max}(P_{\mathrm{c}}) \leq \epsilon_\mathrm{max}$.

    This means that the path $P_{\mathrm{c}}$ exists entirely within the subgraph used to search for $g_{\mathrm{ec}}$ up to the threshold $\epsilon_\mathrm{max}$.
    The existence of such a path ensures that $g_{\mathrm{ec}}$ is defined.
    Thus, the conclusion from Theorem~\ref{theorem:extra-cluster-gap-wo-cg-general} applies.

\end{proof}

Theorem~\ref{theorem:extra-cluster-gap-wo-cg-special} guarantees that the extra-cluster gap w/o CG method can identify every instance where the cluster gap $g_{\text{c}}$ is below a given threshold $\epsilon_\mathrm{max}$.
This property is valuable for applications like decoder switching~\cite{toshio2025decoderswitching} or methods which rely on flagging low-confidence results for further processing~\cite{smith2024mitigatingerrorsinlq,bombin2024faulttolerantmsp,jones2024improvedaccuracy,meister2024efficientsoftoutput,gidney2024msc,hirano2025efficientmsc,dinca2025errormitigation,zhou2025errormitigation,sunami2025entanglementboosting}.
For such methods, it is crucial not to miss any samples below the threshold, making the extra-cluster gap w/o CG a suitable candidate.

\subsubsection{Extra-Cluster Gap with Cluster Graph (w/ CG)}
A limitation of the w/o CG method is that it can yield a value $g_{\text{ec}}\leq\epsilon_\mathrm{max}$ even when the cluster gap is larger, $g_{\text{c}}>\epsilon_\mathrm{max}$.
To address this inaccuracy, we introduce the extra-cluster gap w/ CG.
This method first performs the same additional growth step to detect a connection.
If a connection is found, it then constructs a cluster graph to calculate the precise distance, as shown in Figure~\ref{fig:extra-cluster-gap-with-cg} and detailed in Algorithms~\ref{alg:extra-cluster-gap-w-cluster-graph} and \ref{alg:cluster-graph}.

The output of this algorithm, which we denote as $g_{\mathrm{eccg}}$, is the extra-cluster gap w/ CG, calculated only if the initial growth check is positive.

\begin{definition}[Extra-Cluster Gap w/ CG: $g_{\mathrm{eccg}}$] \label{def:g_eccg}
    The extra-cluster gap with a cluster graph, $g_{\mathrm{eccg}}$, is conditionally calculated.

    If a path between $b_1$ and $b_2$ exists within $G'_{\epsilon_\mathrm{max}}$, then $g_{\mathrm{eccg}}$ is defined as the shortest path distance between $b_1$ and $b_2$ in that subgraph.
    Otherwise, $g_{\mathrm{eccg}}$ is undefined.

\end{definition}

The properties of this method are formalized in the following theorems.

\begin{theorem} \label{theorem:extra-cluster-gap-w-cg-general}
    For any threshold $\epsilon_\mathrm{max} \geq 0$, one of two conditions must hold:
    \begin{itemize}
        \item $g_{\mathrm{eccg}}$ is defined, and it satisfies the inequality $g_{\mathrm{c}}\leq g_{\mathrm{eccg}}$.
        \item $g_{\mathrm{eccg}}$ is undefined.
    \end{itemize}

\end{theorem}
\begin{proof}
    The value $g_{\mathrm{eccg}}$ is defined as the shortest path distance in the subgraph $G'_{\epsilon_\mathrm{max}}$, while $g_{\mathrm{c}}$ is the shortest path distance in the full graph $G'$.
    The subgraph $G'_{\epsilon_\mathrm{max}}$ contains a subset of the edges available in $G'$.
    Assuming non-negative edge weights, the shortest path distance in a larger graph cannot be greater than the shortest path distance in its subgraph.

    Therefore, the shortest path in $G'$ must be less than or equal to the shortest path in $G'_{\epsilon_\mathrm{max}}$, which gives the inequality $g_{\mathrm{c}}\leq g_{\mathrm{eccg}}$.
    This holds whenever $g_{\mathrm{eccg}}$ is defined; otherwise, the second condition is met.

\end{proof}

\begin{figure}[t]
    \centering
    \includegraphics[scale=0.5]{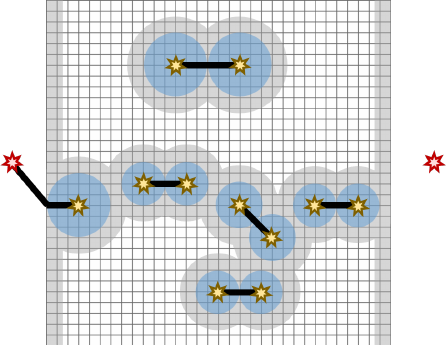}
    \includegraphics[scale=0.5]{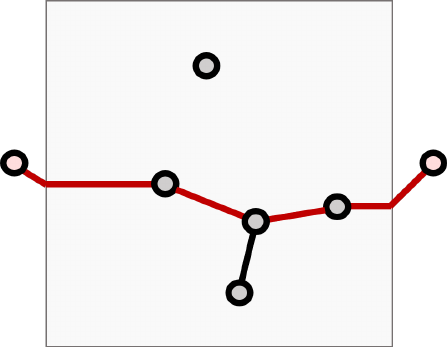}
    \caption{
        ({\bf Left}) A connection between the two boundaries is formed and detected via the extra-cluster gap method.
        ({\bf Right}) A cluster graph is constructed using the distances between the colliding clusters to compute the precise shortest path.
    }
    \label{fig:extra-cluster-gap-with-cg}
\end{figure}

\begin{algorithm}[t]
    \caption{Extra-cluster gap with the cluster graph}
    \label{alg:extra-cluster-gap-w-cluster-graph}
    \begin{algorithmic}[1]
        \Require A graph $G$ with clusters and boundaries $b_1$, $b_2$.
        \Ensure A soft-output value if boundaries are connected by additional growth up to $\epsilon_\mathrm{max}$, otherwise null.
        \While {$\epsilon \leq \epsilon_\mathrm{max}$}
            \State Increase the radii of all clusters and boundary nodes by $\delta\epsilon/2$ ($\delta\epsilon$: arbitrary).
            \State Merge colliding clusters using the Union operation and record the current $\epsilon$ as the collision distance.
            \State $\epsilon\gets \epsilon + \delta\epsilon$
        \EndWhile
        \If {a single cluster connects boundaries $b_1$ and $b_2$}
            \State \Return CLUSTER\_GRAPH()
        \Else
            \State \Return null
        \EndIf
    \end{algorithmic}
\end{algorithm}

\begin{algorithm}[t]
    \caption{Calculation of the distance using a cluster graph}
    \label{alg:cluster-graph}
    \begin{algorithmic}[1]
        \Function {cluster\_graph}{}
            \State Construct a graph (the cluster graph) where nodes are the clusters and edge weights are the collision distances from the additional growth.
            \State \Return The shortest path distance between the boundary nodes on the cluster graph, found using Dijkstra's algorithm.
        \EndFunction
    \end{algorithmic}
\end{algorithm}

\begin{theorem} \label{theorem:extra-cluster-gap-w-cg-special}
    If the cluster gap $g_{\mathrm{c}}$ is less than or equal to the threshold $\epsilon_\text{max}$, then $g_{\mathrm{eccg}}$ is defined and is exactly equal to $g_{\mathrm{c}}$.

\end{theorem}
\begin{proof}
    From Theorem~\ref{theorem:extra-cluster-gap-w-cg-general}, we have the relation $g_{\mathrm{c}}\leq g_{\mathrm{eccg}}$ when $g_{\mathrm{eccg}}$ is defined.
    To prove equality, we must show the reverse inequality, $g_{\mathrm{eccg}} \leq g_{\mathrm{c}}$, under the given condition.

    As shown in the proof of Theorem~\ref{theorem:extra-cluster-gap-wo-cg-special}, $w_\mathrm{max}(P_{\mathrm{c}}) \leq \epsilon_\mathrm{max}$.
    This implies that the entire path $P_{\mathrm{c}}$ is contained within the subgraph $G'_{\epsilon_\mathrm{max}}$.
    Because $P_{\mathrm{c}}$ is a path in $G'_{\epsilon_\mathrm{max}}$, $g_{\mathrm{eccg}}$ must be defined.
    Furthermore, since $g_{\mathrm{eccg}}$ is the length of the shortest path in $G'_{\epsilon_\mathrm{max}}$, it must be less than or equal to the length of any other path in $G'_{\epsilon_\mathrm{max}}$, including $P_{\mathrm{c}}$.
    Thus, we have $g_{\mathrm{eccg}} \leq g_{\mathrm{c}}$.

    Combining the two inequalities, $g_{\mathrm{c}}\leq g_{\mathrm{eccg}}$ and $g_{\mathrm{eccg}} \leq g_{\mathrm{c}}$, we conclude that $g_{\mathrm{c}}= g_{\mathrm{eccg}}$.

\end{proof}

Theorem~\ref{theorem:extra-cluster-gap-w-cg-special} guarantees that the extra-cluster gap w/ CG is exactly equal to the cluster gap for all instances where $g_{\text{c}}\leq\epsilon_\mathrm{max}$.
This makes the method both accurate and efficient, as the expensive calculation is performed only when necessary.

\subsubsection{Implementation Costs}
Finally, we consider the implementation costs of these extra-cluster gap methods.
The additional growth step is nearly identical in implementation to a standard UF decoder, sharing the same time complexity of $O(d^3\alpha(d^3))$, where $d$ is the code distance.
Hardware implementations of UF decoders can achieve decoding times under 1~\textmu{}s for $d=17$ on circuit-level noise models~\cite{liyanage2024heliosv2,ziad2024localclusteringdecoder}.
We expect that our extra-cluster gap w/o CG method can achieve comparable time complexity for similar code distances.
This approach is particularly advantageous in hardware implementation that executes cluster growth in parallel, especially if the additional growth range $\epsilon_\mathrm{max}/2$ is small.
In the next section, we will numerically evaluate these growth ranges and quantify how effectively the w/o CG method minimizes incorrect estimations.

The extra-cluster gap w/ CG method involves an additional step: calculating the shortest path on the cluster graph.
This step is distinct from the standard UF algorithm and adds complexity to a hardware implementation.
However, the probability of forming a boundary-to-boundary connection in a UF decoder is known to decrease rapidly as the code distance increases~\cite{griffiths2024ufwouf}.
We anticipate that such connections will also be rare in our additional growth step.
If these events are infrequent, the computationally intensive cluster graph analysis can be offloaded to separate, specialized hardware, thus minimizing the burden on the primary decoder.
In the next section, we will numerically evaluate the frequency of these connection events.

\section{Numerical Results}\label{sec:numerical}

In this section, we present numerical experiments to evaluate the performance of the bounded cluster gap and the extra-cluster gap.
We performed noisy circuit simulations using Stim~\cite{gidney2021stim} with a circuit-level noise model.
The simulations assumed rotated surface codes with a depth-6 syndrome measurement circuit and a physical error probability $p$.

A UF decoder implemented in Rust was used for our decoding.
From the resulting clusters, we calculated the cluster gap, bounded cluster gap, and extra-cluster gap.
Following previous works~\cite{gidney2025yokedsurfacecodes,jones2024improvedaccuracy,toshio2025decoderswitching}, we express the gap in decibels (dB).
The early-stopping threshold is set to $\epsilon_\mathrm{max}=20$~dB.
This threshold is chosen because it approximates the performance of the strong decoder in a decoder switching~\cite{toshio2025decoderswitching} and serves as a reference in the libra decoder~\cite{jones2024improvedaccuracy}.
The soft outputs obtained from these numerical experiments are consistent with Theorems~\ref{theorem:extra-cluster-gap-wo-cg-general}--\ref{theorem:extra-cluster-gap-w-cg-special}.
A detailed demonstration of this consistency is provided in Appendix~\ref{sec:consistency}.

\subsection{Visited Nodes of Bounded Cluster Gap}

The number of visited nodes serves as a direct proxy for the computational cost.
We therefore compare this metric to assess the performance of our proposed method.
Figure~\ref{fig:actually-visited-nodes} illustrates the reduction in the number of visited nodes when using the bounded cluster gap, which employs an early-stopping Dijkstra's algorithm, compared to the cluster gap.
This difference is more pronounced at lower physical error probabilities $p$.
For instance, the number of visited nodes is reduced by a factor of approximately 100 at $p=0.05\%$ and by a factor of 10 at $p=0.10\%$.

We fit the number of visited nodes for both methods to the power-law function
\begin{align}
    A d^B, \label{eq:nodes-fitting}
\end{align}
where the parameters $A$ and $B$ are determined by a least-squares fit on a log-log plot.
The resulting values of the exponent $B$ are listed in Table~\ref{tab:actually-visited-nodes-params}.
For the bounded cluster gap, the exponent $B$ is small at low physical error probabilities.
At $p=0.10\%$, the scaling is nearly quadratic ($B \approx 2.31$), approaching the complexity outlined in \eqref{eq:bounded-cluster-gap-time-complexity}.
In contrast, the cluster gap exhibits approximately cubic scaling ($B \approx 2.88$).

As $p$ increases, the number of visited nodes increases for the bounded cluster gap but decreases for the cluster gap.
This behavior in the bounded cluster gap occurs because a higher $p$ leads to larger clusters with zero-weight edges.
Even with a fixed $\epsilon_\mathrm{max}$, the algorithm must explore more nodes within these expanded zero-weight regions.
Conversely, for the cluster gap, these zero-weight regions are explored preferentially by Dijkstra's algorithm, allowing it to reach the boundary nodes more quickly and thus reducing the total number of visited nodes.
The number of visited nodes for both methods becomes comparable around $p=1.00\%$.

\begin{figure}[t]
    \includegraphics[scale=0.6]{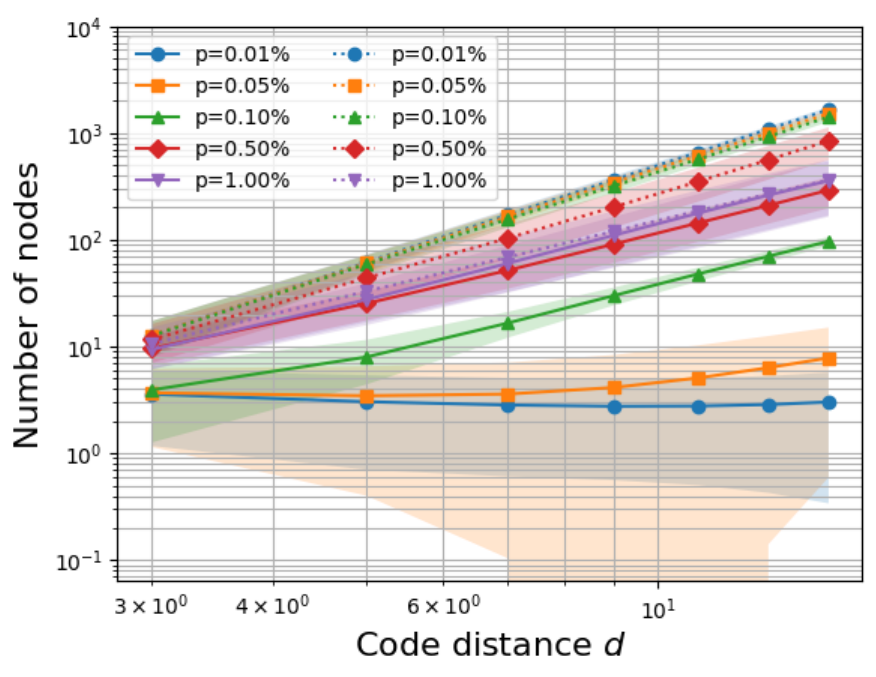}
    \caption{
        The number of nodes visited during Dijkstra's algorithm for the bounded cluster gap (solid lines) and the original cluster gap (dotted lines).
        The shaded areas represent the standard deviation.
        Each data point is an average over $5\cdot 10^6$ samples.
        Samples with no detection events are excluded from the analysis.
    }
    \label{fig:actually-visited-nodes}
\end{figure}

\begin{table}[t]
    \begin{center}
        \caption{
            Fitted exponent $B$ from the power-law fit of the number of visited nodes in Figure~\ref{fig:actually-visited-nodes} to \eqref{eq:nodes-fitting}.
            The fit uses data for code distances $d\geq 7$ to mitigate finite-size effects.
        }
        \label{tab:actually-visited-nodes-params}
        \begin{tabular}{ccc}
            \hline\hline
            $p$ & bounded cluster gap & cluster gap \\
            \hline
            $0.01\%$ & $0.08$ & $2.98$ \\
            $0.05\%$ & $1.03$ & $2.91$ \\
            $0.10\%$ & $2.31$ & $2.88$ \\
            $0.50\%$ & $2.27$ & $2.75$ \\
            $1.00\%$ & $2.34$ & $2.18$ \\
            \hline\hline
        \end{tabular}
    \end{center}
\end{table}

In the very low error regime of $p=0.01\%$, the corresponding edge weight is $w=\ln((1-p)/p)\approx 9.21$.
This value is much larger than the early-stopping threshold, which corresponds to $\epsilon_\mathrm{max}=20\ \mathrm{dB}\approx 4.605$ in natural units.
Consequently, the search terminates before even a single non-zero weight edge can be traversed.
Therefore, the number of visited nodes barely increases with the code distance $d$.

\subsection{Performance of Extra-Cluster Gap}

When a cluster graph is not used, it is possible for a sample to have an extra-cluster gap below $\epsilon_\mathrm{max}$ while its cluster gap is above $\epsilon_\mathrm{max}$.
Figure~\ref{fig:too-rejecting-ratio} plots the fraction of samples where the soft output (either the cluster gap or the extra-cluster gap w/o CG) is less than or equal to $\epsilon_\mathrm{max}=20$~dB.
For $p\leq 0.10\%$, this fraction decreases exponentially with $d$ for the extra-cluster gap w/o CG, similar to the trend observed for the cluster gap.
This exponential decay, mirroring the behavior of the cluster gap, confirms the practical viability of using the extra-cluster gap for applications such as decoder switching~\cite{toshio2025decoderswitching}.
However, for higher error rates ($p\geq 0.50\%$), this fraction no longer decreases for the extra-cluster gap, in contrast to the cluster gap, which still shows a slight decrease.

\begin{figure}[t]
    \includegraphics[scale=0.6]{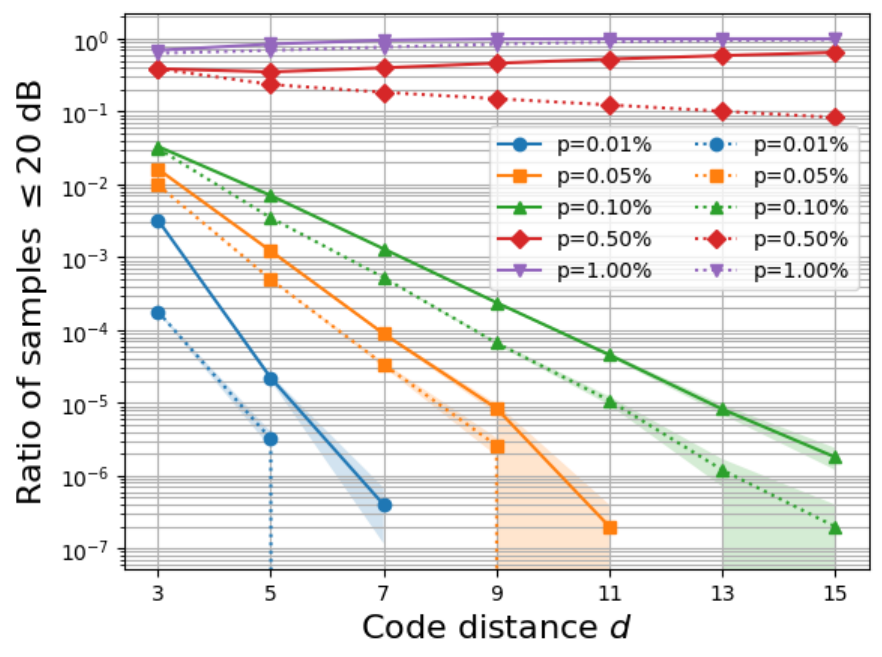}
    \caption{
        Comparison of the fraction of samples with a soft output below the threshold $\epsilon_\mathrm{max}=20$~dB.
        The solid line represents the extra-cluster gap w/o CG, and the dotted line represents the cluster gap.
        The shaded areas indicate the standard error.
        Each data point is an average over $5\cdot 10^6$ samples.
    }
    \label{fig:too-rejecting-ratio}
\end{figure}

\begin{table}[t]
    \begin{center}
        \caption{
            Fitted exponent B from the exponential fit of the data in Figure~\ref{fig:too-rejecting-ratio} to Eq.\eqref{eq:prob-fitting}.
            The table shows results only for $p\leq 0.10\%$, where the extra-cluster gap exhibits a negative slope.
            At $p=0.10\%$, the prefactor for the extra-cluster gap w/o CG is $A=10^{-0.38}$.
        }
        \label{tab:too-rejecting-ratio-params}
        \begin{tabular}{ccc}
            \hline\hline
            $p$ & extra-cluster gap w/o CG & cluster gap \\
            \hline
            $0.01\%$ & $-0.98$ & $-0.87$ \\
            $0.05\%$ & $-0.60$ & $-0.60$ \\
            $0.10\%$ & $-0.36$ & $-0.43$ \\
            \hline\hline
        \end{tabular}
    \end{center}
\end{table}

We fit the data to the exponential function
\begin{align}
    A \cdot 10^{Bd} \label{eq:prob-fitting}
\end{align}
using a least-squares method on a semi-log plot. The resulting exponents $B$ are presented in Table~\ref{tab:too-rejecting-ratio-params}.

The physical interpretation of this fraction depends on the application.
When a cluster graph is used, this fraction represents the probability that calculates the shortest distance on the graph is necessary.
Without a cluster graph, it corresponds to the post-selection rate in certain fault-tolerant schemes~\cite{bombin2024faulttolerantmsp,smith2024mitigatingerrorsinlq,gidney2024msc,hirano2025efficientmsc,dinca2025errormitigation,zhou2025errormitigation,sunami2025entanglementboosting} or the switching rate in hybrid decoders like decoder switching~\cite{toshio2025decoderswitching} and libra~\cite{jones2024improvedaccuracy}.
In Section~\ref{sec:decoder-switching}, we will focus on the implications for decoder switching.

\begin{figure}[t]
    \includegraphics[scale=0.6]{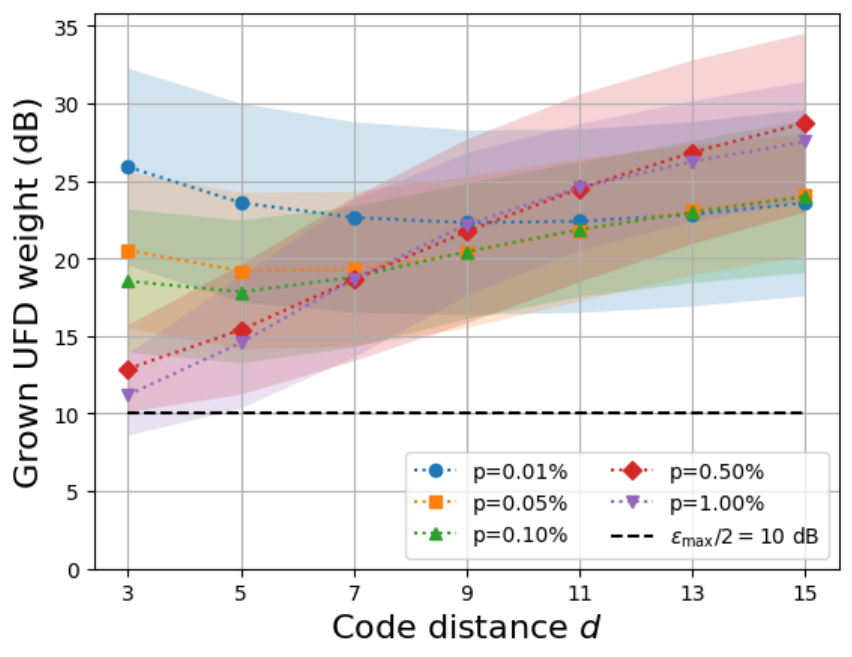}
    \caption{
        The maximum growth radius required for the UF decoder to complete its search.
        The shaded areas represent the standard deviation.
        Each data point is an average over $2\cdot 10^6$ samples.
    }
    \label{fig:ufd-weight}
\end{figure}

For parallel hardware implementations of a UF decoder, such as on an FPGA, growth operations at each node can be performed concurrently~\cite{liyanage2024heliosv2,chan2023actisstrictlylocal,ziad2024localclusteringdecoder}.
In this context, a key factor determining the total computation time is the number of parallel growth iterations required for the algorithm to terminate.
Figure~\ref{fig:ufd-weight} shows the maximum growth radius required for the standard UF decoder to complete.
The extra-cluster gap calculation limits this growth to a fixed value of $\epsilon_\mathrm{max}/2=10$~dB.
In contrast, the standard UF decoder requires a growth radius exceeding 20~dB for all tested $p$ and $d\geq 9$.
This suggests that calculating the extra-cluster gap requires fewer growth iterations than a full UF decoding, which could lead to a reduction in computation time in a parallel implementation.

For a sequential implementation, the total number of nodes within all clusters is a more relevant metric for computational cost~\cite{delfosse2020lineartimemaximumlikelihood}; these results are presented in Appendix~\ref{sec:nodes-in-clusters}.
As detailed in the appendix, for $p\leq 0.10\%$, the additional cluster growth for the extra-cluster gap results in a number of cluster nodes that scales more favorably with code distance $d$ compared to the standard UF decoder.

It is also insightful to compare the computational costs of the original cluster gap and the extra-cluster gap w/o CG.
A direct comparison is challenging because they rely on fundamentally different algorithms: the former uses a Dijkstra's search, while the latter employs an additional cluster growth.
Nevertheless, examining the number of nodes involved in each process provides a useful point of reference.
For instance, at a physical error rate of $p = 0.10\%$, the number of additional nodes engaged by the extra-cluster gap calculation (Figure~\ref{fig:nodes-within-clusters}, bottom) is smaller than the number of nodes visited by the cluster gap algorithm (Figure~\ref{fig:actually-visited-nodes}).

\section{Applications of Extra-Cluster Gap}\label{sec:applications}

In this section, we apply the results from our numerical experiments to evaluate the performance of our early-stopping techniques in several quantum error correction (QEC) applications.

\subsection{Decoder Switching}\label{sec:decoder-switching}

We now evaluate the performance of the extra-cluster gap w/o CG from the perspective of the decoder-switching scheme~\cite{toshio2025decoderswitching}.
Specifically, we investigate whether this method can prevent the backlog problem in two different scenarios.
The first scenario involves small code distances ($d \leq 17$), which are relevant for near-term quantum computers.
The second considers a practical code distance of $d=25$, which is required for large-scale applications such as 2048-bit factorization~\cite{gidney2025howtofactor}.

For both scenarios, we assume a physical error probability of $p=0.10\%$ under a circuit-level noise model and a syndrome generation time of $\tau_\mathrm{gen}=1$~\textmu{}s.
Following the setup in Ref.~\cite{toshio2025decoderswitching}, we assume the communication time for the weak decoder is equal to $\tau_\mathrm{gen}$, while both the decoding and communication times for the strong decoder are $10\tau_\mathrm{gen}$.

First, let us consider the near-term scenario with $d \leq 17$.
The results in Figure~\ref{fig:ufd-weight} show that the number of growth iterations required for the extra-cluster gap w/o CG is approximately half that of a full UF decoder.
Based on this, we make a pessimistic estimate that the computation time for the weak decoder is $\tau_\mathrm{dec}^\mathrm{weak}\approx 2\tau_\mathrm{UFD}$, where $\tau_\mathrm{UFD}$ is the computation time of the UF decoder.
For code distances up to $d=17$, the computation time of a UF decoder is reported to be at most $\tau_\mathrm{UFD}\leq 0.045$~\textmu{}s~\cite{liyanage2024heliosv2}, which gives
\begin{align}
    \tau_\mathrm{dec}^\mathrm{weak} / \tau_\mathrm{gen} \simeq 0.09.
\end{align}
According to Theorem~1 in Ref.~\cite{toshio2025decoderswitching}, for these setups, a backlog problem is expected to occur if the switching rate exceeds approximately $5 \times 10^{-2}$.
In our case, the switching rate corresponds to the probability  that $g_\mathrm{eccg}$ falls below 20~dB.
Then, Figure~\ref{fig:too-rejecting-ratio} indicates that the switching rate is at most $2 \times 10^{-2}$ for our setups.
Since this rate is well below the theoretical bound, we conclude that decoder switching using the extra-cluster gap w/o CG can successfully avoid the backlog problem even for code distances up to $d \leq 17$.

Next, we consider the large-scale application scenario with $d=25$.
For such a large code distance, a fully parallel implementation of the UF decoder, where each node of the decoding graph is mapped to a dedicated Processing Element (PE), becomes infeasible due to resource limitations.
To address this issue, time-multiplexing can be employed, where a single PE handles multiple nodes sequentially~\cite{liyanage2024heliosv2,ziad2024localclusteringdecoder}.
This approach reduces the required number of Look-Up Tables (\#LUTs) by a factor of approximately $1/n$ at the cost of increasing the computation time by a factor of $n$, where $n$ is the multiplexing factor.

Without time-multiplexing, a $d=25$ implementation is estimated to require approximately $3.7 \times 10^6$ LUTs (see Appendix~\ref{sec:luts} for details).
To fit within the resource constraints outlined in Table~1 of Ref.~\cite{liyanage2024heliosv2}, a time-multiplexing factor of $n=5$ is necessary.
Although Ref.~\cite{liyanage2024heliosv2} indicates that the UF decoder's execution time per round tends to decrease with increasing code distance, we adopt a pessimistic assumption.
We take the time for $d=25$ to be approximately 0.025~\textmu{}s, the value reported for $d=17$.
With a time-multiplexing factor of $n=5$, the UF decoder computation time becomes $\tau_\mathrm{UFD} = 5 \times 0.025~\text{\textmu s} = 0.125$~\textmu{}s.
Consequently, the weak decoder computation time, including the extra-cluster gap calculation, is estimated as $\tau_\mathrm{dec}^\mathrm{weak} = 2\tau_\mathrm{UFD} = 0.25$~\textmu{}s, leading to
\begin{align}
    \tau_\mathrm{dec}^\mathrm{weak} / \tau_\mathrm{gen} \simeq 0.25.
\end{align}
In this configuration, the backlog problem arises if the switching rate exceeds approximately $4 \times 10^{-2}$.
According to Table~\ref{tab:too-rejecting-ratio-params}, the switching rate for $d=25$ when using the extra-cluster gap w/o CG is merely $4.17 \times 10^{-10}$, which is orders of magnitude lower than the threshold.
Therefore, even for a large code distance of $d=25$, our proposed decoder-switching scheme can easily avoid the backlog problem.

In summary, our analysis shows that a decoder switching scheme incorporating the extra-cluster gap w/o CG enables backlog-free decoding across a wide range of scenarios.
This includes surface codes of sizes that will be feasible in the near future, as well as large-scale codes that will be required for practical applications in the FTQC era.

\subsection{Multiple Logical Boundaries}\label{sec:multiple-logical}

Next, we analyze the performance of soft-output computation in the presence of multiple logical boundaries.
Such configurations arise in architectures that use lattice surgery to perform entangling gates between logical qubits~\cite{horsman2012latticesurgery}.
For example, Figure~\ref{fig:multiple-logical-qubits} shows the \textit{compact-block} layout from Ref.~\cite{litinski2019gameofsurfacecodes}, where logical qubits $\ket{q_1}, \ldots, \ket{q_{12}}$ are coupled via a single ancilla region. 
Decoding such large-scale QEC codes often involves spatial partitioning with buffer zones of width $d$~\cite{lin2025spatiallyparalleldecoding,bombin2023modulardecoding}.
This partitioning yields multiple decoding problems within a certain sub-region, like the one outlined in blue in Figure~\ref{fig:multiple-logical-qubits}.
This blue region contains eight distinct $X$ boundaries, leading to $\binom{8}{2}=28$ possible pairings for which a soft output might be calculated.

\begin{figure}[t]
    \includegraphics[scale=0.7]{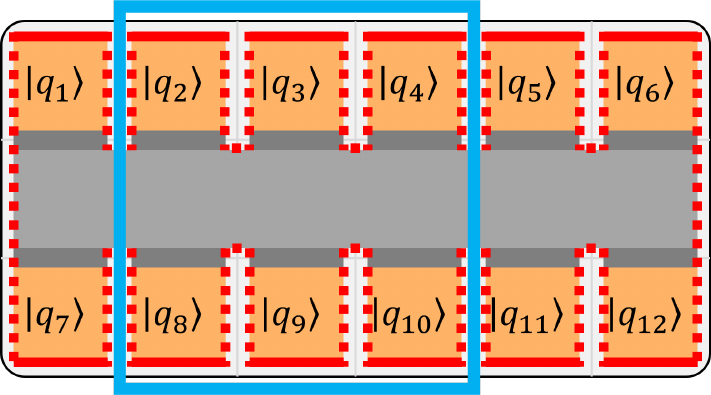}
    \caption{
        A schematic of multiple surface code patches forming the \textit{compact-block} layout in Ref.~\cite{litinski2019gameofsurfacecodes}.
        Solid red lines denote $Z$ boundaries, and dashed red lines denote $X$ boundaries.
        The blue box outlines a decoding region, including buffer zones, in a spatially parallel decoding scheme~\cite{lin2025spatiallyparalleldecoding}.
    }
    \label{fig:multiple-logical-qubits}
\end{figure}

The computational cost for this setup depends heavily on which type of soft output is employed. For example, when using the complementary gap, we need to solve separate MLE decoding tasks for each of the 28 pairs.
As suggested in prior works~\cite{meister2024efficientsoftoutput,lee2025efficientpostselection}, the decoding process becomes more efficient for the cluster gap or bounded cluster gap.
However, even in these cases, Dijkstra's search is still required from each of the eight boundaries.
In contrast to these previous attempts, the extra-cluster gap w/o CG only requires a single cluster growth operation.
When using a cluster graph (CG), the results in Table~\ref{tab:too-rejecting-ratio-params} indicate that the probability of needing a cluster graph calculation is low for $p=0.10\%$ and large $d$, since $28 \cdot 10^{-0.36d} \ll 1$.

\begin{table}[t]
    \begin{center}
        \caption{
            Expected number of soft-output calculations for a system with $M$ non-equivalent logical boundaries, assuming $p=0.10\%$ and sufficiently large $d$.
        }
        \begin{tabular}{lc}
            \hline\hline
            Method & Expected Computations \\
            \hline
            complementary gap & $O(M^2)$ \\
            cluster gap & $O(M)$ \\
            bounded cluster gap & $O(M)$ \\
            extra-cluster gap w/o CG & 1 \\
            extra-cluster gap w/ CG & $O(1)$ \\
            \hline\hline
        \end{tabular}
        \label{tab:expected-number-in-multi}
    \end{center}
\end{table}

More generally, for a system with $M$ logical boundaries, which can arise from partitioning in both space and time~\cite{dennis2002topologicalquantummemory,skoric2023parallelwindowdecoding,liyanage2025heliosv3,bombin2023modulardecoding}, the expected number of soft-output computations for each method scales as shown in Table~\ref{tab:expected-number-in-multi}.
These results demonstrate that the extra-cluster gap enables fast and scalable soft-output computation even in complex architectures with many logical boundaries.
This advantage is particularly relevant for general qLDPC codes, which can encode multiple logical qubits and for which the complementary gap is often impractical~\cite{lee2025efficientpostselection}.
The extra-cluster gap is therefore a promising tool for use with cluster-based decoders for qLDPC codes~\cite{wolanski2024ambiguityclustering,delfosse2022ufqldpc}.

\section{Conclusion}\label{sec:conclusion}

In this work, we introduced early-stopping techniques to accelerate the computation of soft outputs for real-time QEC decoding.
Specifically, we proposed two specific methods: the bounded cluster gap, which employs a bounded Dijkstra's algorithm, and the extra-cluster gap, which computes a soft output from minimally grown clusters.

Our analysis shows that the bounded cluster gap and the extra-cluster gap w/ CG produce results identical to the original cluster gap for all soft outputs below a predefined threshold $\epsilon_{\max}$.
This allows for significant computational speedups while preserving the performance benefits of the cluster gap, such as its use in post-selection.
The extra-cluster gap w/o CG is particularly well-suited for hardware implementation, as it reuses the standard cluster growth module of decoders like Union-Find decoder.
Crucially, this method does not miss any samples where the cluster gap is below $\epsilon_{\max}$.

Numerical experiments at a physical error rate of $p=0.10\%$ revealed that the bounded cluster gap exhibits a more favorable polynomial scaling with code distance $d$ compared to the original cluster gap.
Furthermore, the extra-cluster gap proved effective for applications such as decoder switching and for scenarios involving multiple logical boundaries, where it offers a significant performance advantage.

Future work will be directed toward implementing these algorithms on FPGAs to experimentally demonstrate the speed advantage of our early-stopping techniques.

\vskip2\baselineskip
NOTE ADDED: While completing this manuscript, we became aware of a related work by Ref.~\cite{xie2026simple}, which pursues a similar goal using a completely different approach and codes.
A key distinction is that their approach requires an additional re-decoding step after graph reweighting, whereas our extra-cluster growth method avoids such a computationally expensive process entirely.

\section{Acknowledgments}

We are grateful to thank Takumi Akiyama, Yugo Takada, Yutaro Akahoshi, Moeto Mishima, Shinichiro Yamano, Mitsuki Katsuda, Hoiki Liu, and Koki Chinzei for fruitful discussions.
K. F. is supported by MEXT Quantum Leap Flagship Program (MEXT Q-LEAP)
Grant No. JPMXS0120319794, JST COI-NEXT Grant No. JPMJPF2014, JST
Moonshot R\&D Grant No. JPMJMS2061, and JST CREST JPMJCR24I3.

\textbf{Author contributions}: R. T. initially conceived the concept of the extra-cluster growth method.
K. K. subsequently proposed its application to the calculation of soft outputs and introduced the early-stopping framework.
K. K. formulated the methods, implemented and performed all numerical simulations, and wrote the original draft of the manuscript.
K. K., R. T., and K. F. collaboratively developed the fundamental aspects of the theoretical proofs, which K. K. then finalized.
R. T. proposed the cluster graph and drafted the schematic illustrations.
J. F., H. O., and S. S. provided overall supervision, environments, and resources for this work and guided the research direction.
K. F. provided technical supervision, contributed to the conceptualization and the interpretation of the numerical results, and suggested the practical utility of the extra-cluster gap without a cluster graph (w/o CG).
All authors discussed the results and reviewed the manuscript.

\appendix

\section{The $\Delta$-stepping Algorithm on FPGAs} \label{sec:delta-stepping-fpga}

This appendix discusses the challenges of implementing the $\Delta$-stepping algorithm for shortest-path calculations on an FPGA alongside a Union-Find (UF) decoder.

A parallel UF decoder requires a number of processing cores proportional to the number of nodes in the decoding graph~\cite{liyanage2024heliosv2}.
In contrast, a parallel implementation of the $\Delta$-stepping algorithm~\cite{meyer2003deltastepping} requires cores proportional to both the number of nodes and the number of edges.
Although the performance of the $\Delta$-stepping algorithm can be improved by precomputing shortcut edges, this precomputation step has a time complexity of $O(\log d)$ and demands significant hardware resources.

More importantly, this precomputation would need to be performed for every sample, since the decoding graph is dynamically modified by the UF decoder, which sets the weights of intra-cluster edges to zero.
This makes precomputation impractical.
Even if shortcut edges are not used, which slows down the algorithm by a constant factor, the hardware requirements for $\Delta$-stepping remain substantial.
Therefore, implementing the $\Delta$-stepping algorithm separately from the UF decoder is challenging on resource-constrained platforms such as FPGAs.

\section{Number of Nodes in Clusters} \label{sec:nodes-in-clusters}

\begin{figure}[t]
    \includegraphics[scale=0.6]{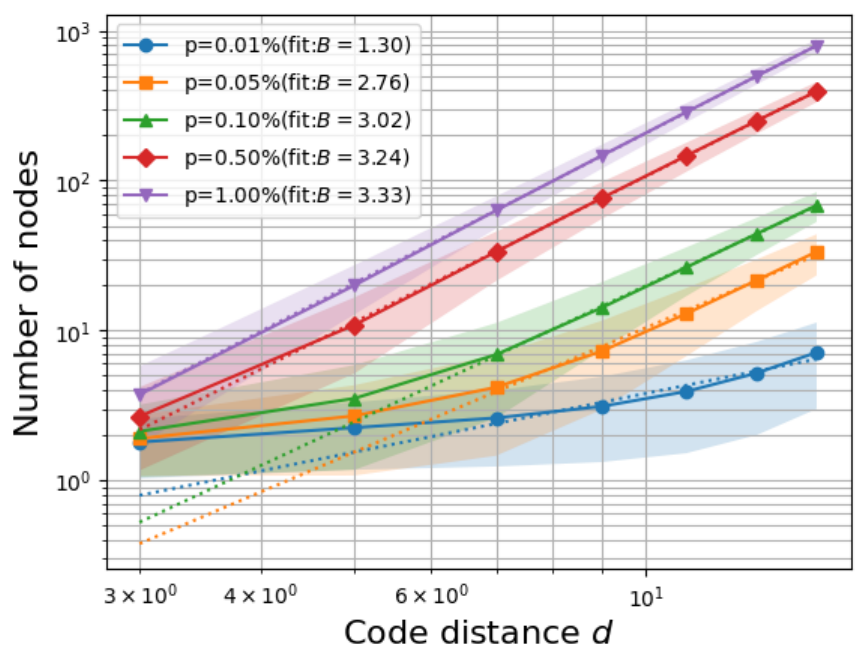}
    \includegraphics[scale=0.6]{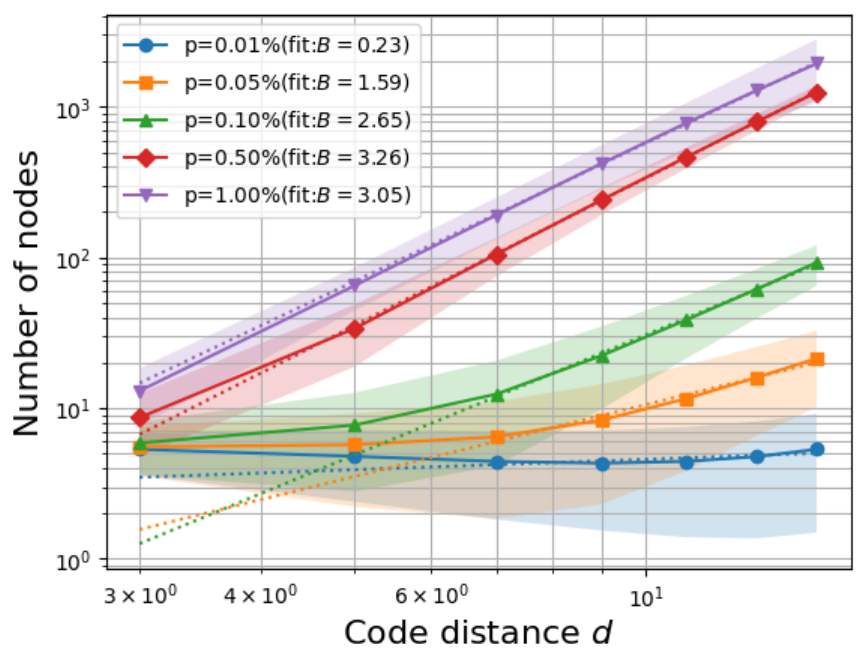}
    \caption{
        Scaling of cluster sizes with code distance $d$ for various physical error rates $p$.
        ({\bf Top}) The total number of nodes within clusters identified by the standard UF decoder.
        ({\bf Bottom}) The number of additional nodes included during the growth step of the extra-cluster gap method.
        Each data point is an average over $10^6$ samples, and the shaded areas represent the standard deviation.
        Dotted lines show the results of a power-law fit to the data for $d \geq 7$, with the fitted exponent $B$ listed in the legend.
    }
    \label{fig:nodes-within-clusters}
\end{figure}

Figure~\ref{fig:nodes-within-clusters} shows the scaling of the number of nodes within clusters as a function of the code distance $d$.
The top panel displays the size of clusters formed by the standard UF decoder, while the bottom panel shows the number of additional nodes incorporated during the growth phase of the extra-cluster gap method.

In both cases, the cluster size grows more rapidly with the code distance $d$ as the physical error probability $p$ increases.
This is expected, as higher error rates lead to larger error clusters.
Notably, for low error rates ($p \leq 0.10\%$), the number of additional nodes from the extra-cluster gap grows with a smaller exponent than the number of nodes in the original clusters.
This indicates a more favorable scaling for the additional growth step required by our method in the low-error regime.

\section{Estimation of Required LUTs for $d=25$}\label{sec:luts}

\begin{figure}[t]
    \includegraphics[scale=0.6]{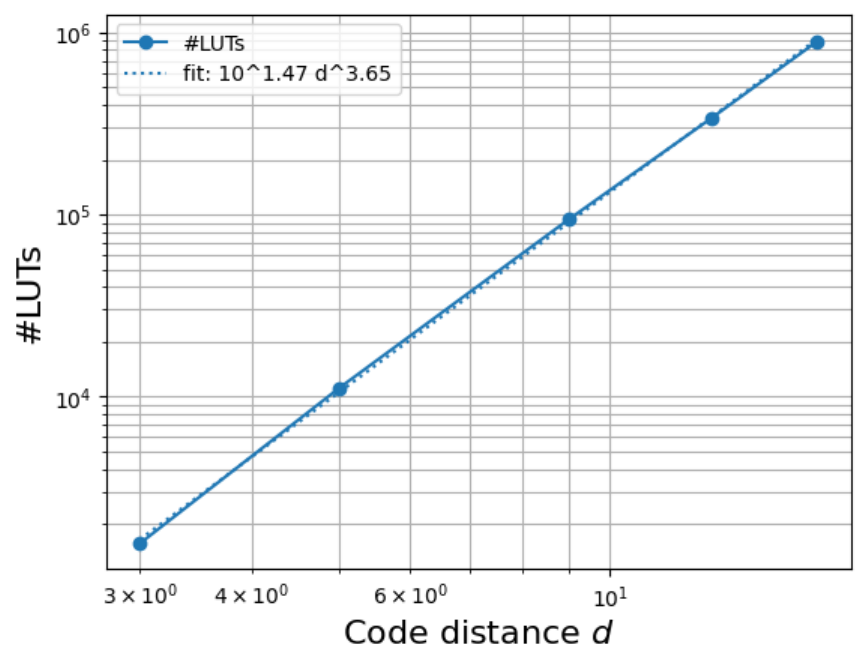}
    \caption{
        The solid line plots the number of Look-Up Tables (\#LUTs) for the circuit-level noise model, with data taken from Table~1 of Ref.~\cite{liyanage2024heliosv2}.
        The dotted line represents a least-squares fit to this data on a log-log scale using \eqref{eq:nodes-fitting}.
    }
    \label{fig:lut-by-helios}
\end{figure}

In this appendix, we estimate the number of LUTs required for an FPGA implementation with a code distance of $d=25$ without time-multiplexing.
Figure~\ref{fig:lut-by-helios} shows the required \#LUTs for various code distances under a circuit-level noise model, as reported in Table~1 of Ref.~\cite{liyanage2024heliosv2}.
By extrapolating from a least-squares fit to this data using \eqref{eq:nodes-fitting}, we find that the estimated number of LUTs required for $d=25$ is $3.7\cdot 10^6$.

\section{Consistency with Early Stopping}\label{sec:consistency}

\begin{figure*}
    \includegraphics[scale=0.15]{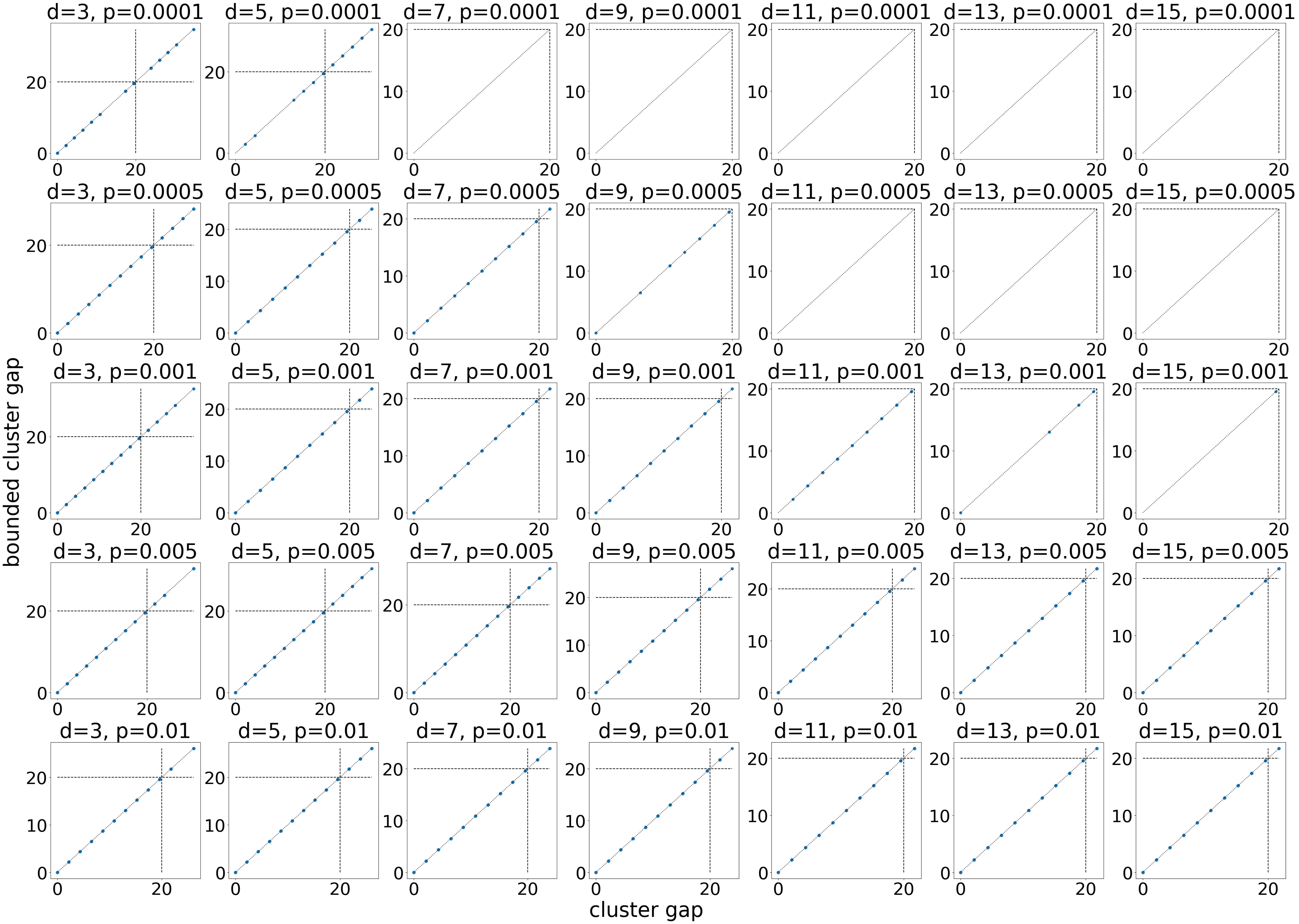}
    \caption{
        Comparison between the bounded cluster gap (vertical axis) and the original cluster gap (horizontal axis) for various code distances $d$ and physical error rates $p$.
        Each plot is generated from $5 \times 10^6$ samples.
        The dashed diagonal line represents equality between the two gaps, confirming their consistency up to the threshold of $\epsilon_\mathrm{max} = 20$~dB.
    }
    \label{fig:consistency-bounded}
\end{figure*}

\begin{figure*}[htb]
    \includegraphics[scale=0.15]{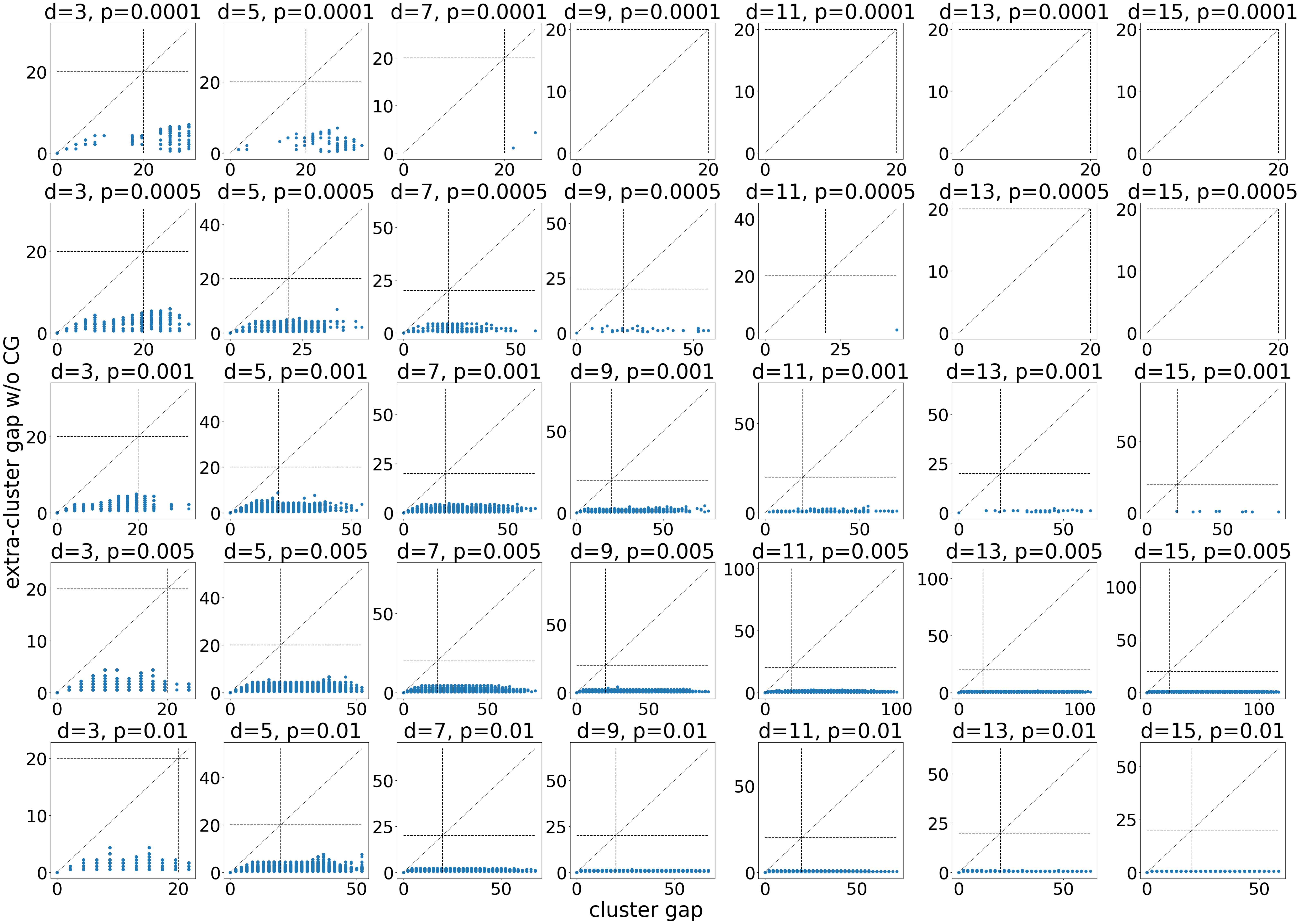}
    \includegraphics[scale=0.15]{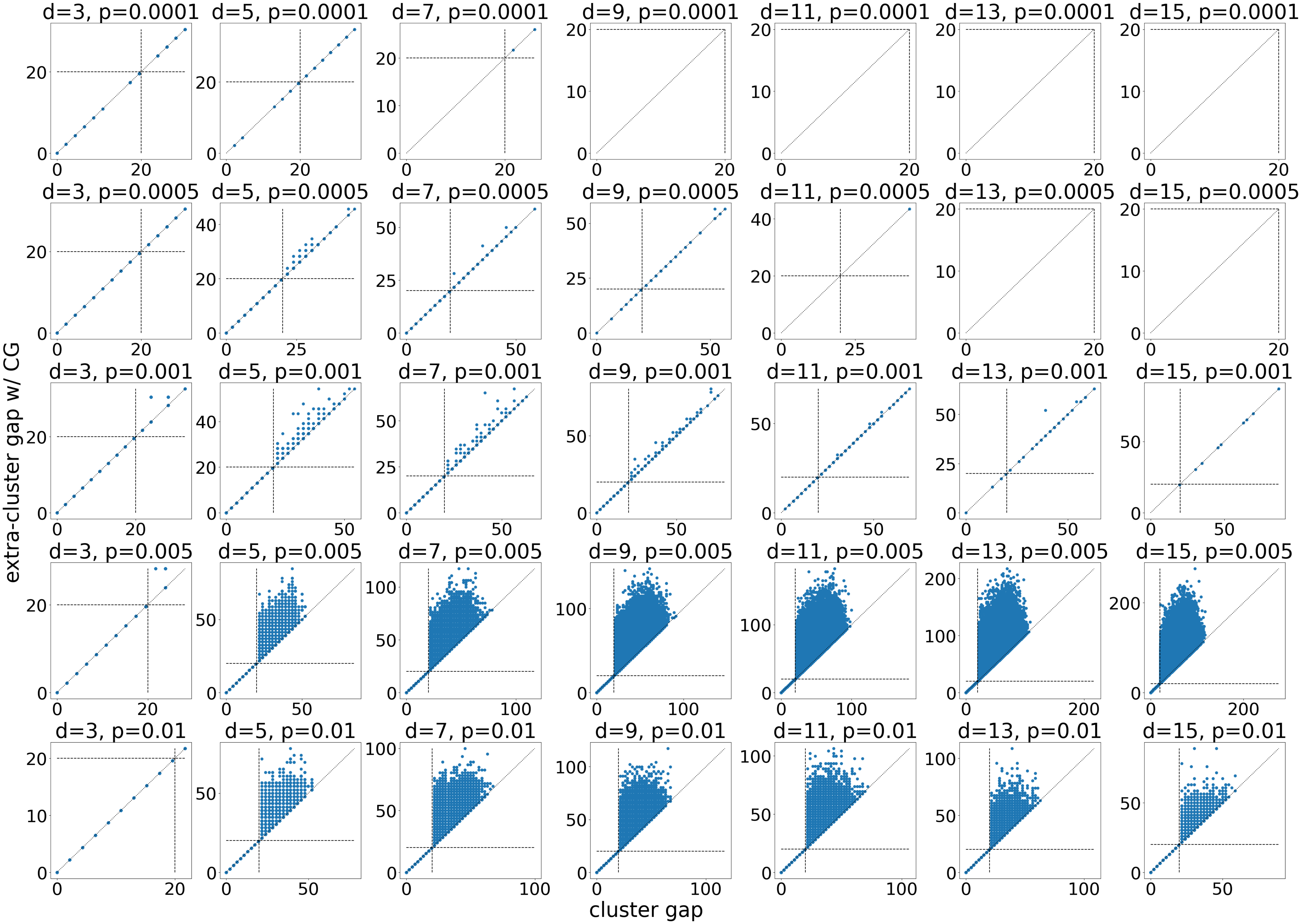}
    \caption{
        Comparison between the extra-cluster gap (vertical axis) and the original cluster gap (horizontal axis) for various code distances $d$ and physical error rates $p$.
        The top row of plots shows the result for the extra-cluster gap without a cluster graph (w/o CG), while the bottom row shows the result with a cluster graph (w/ CG).
        Each plot is generated from $5 \times 10^6$ samples.
        The dashed diagonal line indicates equality between the plotted quantities.
    }
    \label{fig:consistency-extra}
\end{figure*}

In this appendix, we verify that the soft-output values obtained from our proposed methods are consistent with the theoretical predictions.

First, we examine the bounded cluster gap, which is introduced in Section~\ref{sec:bounded-cluster-gap}.
This method is designed to identify all cluster gaps with a value up to a predefined threshold, $\epsilon_\mathrm{max}$, by leveraging the properties of the bounded Dijkstra's algorithm~\cite{bemten2019boundeddijkstra}.
As depicted in Figure~\ref{fig:consistency-bounded}, our numerical results confirm this behavior.
The values of the bounded cluster gap and the original cluster gap are in perfect agreement for all samples with a gap up to the threshold of $\epsilon_\mathrm{max} = 20$~dB.
Furthermore, our results show that a soft-output value is always produced whenever the cluster gap is less than or equal to $\epsilon_\mathrm{max}$, ensuring no instances are missed.

Next, we evaluate the consistency of the extra-cluster gap, which, as detailed in Section~\ref{sec:extra-cluster-gap}, has two variants: one without a cluster graph (w/o CG) and another with a cluster graph (w/ CG).
The theoretical behavior of these two variants differs.

According to Theorems~\ref{theorem:extra-cluster-gap-wo-cg-general} and \ref{theorem:extra-cluster-gap-wo-cg-special}, the extra-cluster gap w/o CG guarantees that no sample with a cluster gap below $\epsilon_\mathrm{max}$ is missed.
However, it may still produce an output below $\epsilon_\mathrm{max}$ even when the cluster gap is larger.
In contrast, Theorems~\ref{theorem:extra-cluster-gap-w-cg-general} and \ref{theorem:extra-cluster-gap-w-cg-special} state that the w/ CG variant is more precise.
It provides a value exactly equal to the cluster gap for all instances up to $\epsilon_\mathrm{max}$ and does not incorrectly report a value below this threshold for samples with a cluster gap larger than $\epsilon_\mathrm{max}$.

The plots in Figure~\ref{fig:consistency-extra} confirm that both the extra-cluster gap w/o CG and w/ CG variants exhibit their respective theoretical behaviors.
For both methods, we also confirmed that a soft output was consistently generated for every sample with a cluster gap below the $\epsilon_\mathrm{max}$ threshold.

These findings collectively confirm that our implementation of the proposed methods aligns with the theoretically predicted outcomes.

\clearpage
\bibliography{bibliography}

@misc{toshio2025decoderswitching,
  title = {Decoder Switching: Breaking the Speed-Accuracy Tradeoff in Real-Time Quantum Error Correction},
  author = {Riki Toshio and Kaito Kishi and Jun Fujisaki and Hirotaka Oshima and Shintaro Sato and Keisuke Fujii},
  year = {2025},
  eprint = {arXiv:2510.25222},
  primaryClass={quant-ph}
}

@misc{meister2024efficientsoftoutput,
  title = {Efficient soft-output decoders for the surface code},
  author = {Nadine Meister and Christopher A. Pattison and John Preskill},
  year = {2024},
  eprint = {arXiv:2405.07433},
  primaryClass={quant-ph}
}

@article{meyer2003deltastepping,
title = {{$\Delta$}-stepping: a parallelizable shortest path algorithm},
journal = {Journal of Algorithms},
volume = {49},
number = {1},
pages = {114-152},
year = {2003},
note = {1998 European Symposium on Algorithms},
issn = {0196-6774},
doi = {https://doi.org/10.1016/S0196-6774(03)00076-2},
url = {https://www.sciencedirect.com/science/article/pii/S0196677403000762},
author = {U. Meyer and P. Sanders}
}

@misc{duan2025fasterdijkstra,
  title = {Breaking the Sorting Barrier for Directed Single-Source Shortest Paths},
  author = {Ran Duan and Jiayi Mao and Xiao Mao and Xinkai Shu and Longhui Yin},
  year = {2025},
  eprint = {arXiv:2504.17033},
  primaryClass={cs.DS}
}

@article{terhal2015backlog,
  title = {Quantum error correction for quantum memories},
  author = {Terhal, Barbara M.},
  journal = {Rev. Mod. Phys.},
  volume = {87},
  issue = {2},
  pages = {307--346},
  numpages = {40},
  year = {2015},
  month = {Apr},
  publisher = {American Physical Society},
  doi = {10.1103/RevModPhys.87.307},
  url = {https://link.aps.org/doi/10.1103/RevModPhys.87.307}
}

@article{gidney2025yokedsurfacecodes,
  author = {Craig Gidney and Michael Newman and Peter Brooks and Cody Jones},
  journal = {Nat. Commun.},
  number = {4498},
  title = {Yoked surface codes},
  volume = {16},
  year = {2025},
  doi = {10.1038/s41467-025-59714-1}
}

@article{bombin2024faulttolerantmsp,
  title = {Fault-Tolerant Postselection for Low-Overhead Magic State Preparation},
  author = {Bomb\'{\i}n, H\'ector and Pant, Mihir and Roberts, Sam and Seetharam, Karthik I.},
  journal = {PRX Quantum},
  volume = {5},
  issue = {1},
  pages = {010302},
  numpages = {19},
  year = {2024},
  month = {Jan},
  publisher = {American Physical Society},
  doi = {10.1103/PRXQuantum.5.010302},
  url = {https://link.aps.org/doi/10.1103/PRXQuantum.5.010302}
}

@article{smith2024mitigatingerrorsinlq,
  author = {Samuel C. Smith and Benjamin J. Brown and Stephen D. Bartlett},
  journal = {Commun. Phys.},
  number = {386},
  title = {Mitigating errors in logical qubits},
  volume = {7},
  year = {2024}
}

@ARTICLE{liyanage2024heliosv2,
  author={Liyanage, Namitha and Wu, Yue and Tagare, Siona and Zhong, Lin},
  journal={IEEE Transactions on Quantum Engineering}, 
  title={FPGA-Based Distributed Union-Find Decoder for Surface Codes}, 
  year={2024},
  volume={5},
  number={},
  pages={1-18},
  doi={10.1109/TQE.2024.3467271}
}

@article{higgott2025sparseblossom,
  doi = {10.22331/q-2025-01-20-1600},
  url = {https://doi.org/10.22331/q-2025-01-20-1600},
  title = {Sparse {B}lossom: correcting a million errors per core second with minimum-weight matching},
  author = {Higgott, Oscar and Gidney, Craig},
  journal = {{Quantum}},
  issn = {2521-327X},
  publisher = {{Verein zur F{\"{o}}rderung des Open Access Publizierens in den Quantenwissenschaften}},
  volume = {9},
  pages = {1600},
  month = jan,
  year = {2025}
}

@INPROCEEDINGS {wu2023fusionblossom,
author = { Wu, Yue and Zhong, Lin },
booktitle = { 2023 IEEE International Conference on Quantum Computing and Engineering (QCE) },
title = {{ Fusion Blossom: Fast MWPM Decoders for QEC }},
year = {2023},
volume = {},
ISSN = {},
pages = {928-938},
doi = {10.1109/QCE57702.2023.00107},
url = {https://doi.ieeecomputersociety.org/10.1109/QCE57702.2023.00107},
publisher = {IEEE Computer Society},
address = {Los Alamitos, CA, USA},
month =sep
}

@article{delfosse2020lineartimemaximumlikelihood,
  title = {Linear-time maximum likelihood decoding of surface codes over the quantum erasure channel},
  author = {Delfosse, Nicolas and Z\'emor, Gilles},
  journal = {Phys. Rev. Res.},
  volume = {2},
  issue = {3},
  pages = {033042},
  numpages = {5},
  year = {2020},
  month = {Jul},
  publisher = {American Physical Society},
  doi = {10.1103/PhysRevResearch.2.033042},
  url = {https://link.aps.org/doi/10.1103/PhysRevResearch.2.033042}
}

@misc{wu2022interpretationunionfind,
  title = {An interpretation of Union-Find Decoder on Weighted Graphs},
  author = {Yue Wu and Namitha Liyanage and Lin Zhong},
  year = {2022},
  eprint = {arXiv:2211.03288},
  primaryClass={quant-ph}
}

@misc{jones2024improvedaccuracy,
  title = {Improved accuracy for decoding surface codes with matching synthesis},
  author = {Cody Jones},
  year = {2024},
  eprint = {arXiv:2408.12135},
  primaryClass={quant-ph}
}

@article{dijkstra1959anote,
  author = {E. W. Dijkstra},
  journal = {Numerische Mathematik},
  pages = {269--271},
  title = {A note on two problems in connexion with graphs},
  volume = {1},
  year = {1959}
}

@misc{bemten2019boundeddijkstra,
  title = {Bounded Dijkstra (BD): Search Space Reduction for Expediting Shortest Path Subroutines},
  author = {Amaury Van Bemten and Jochen W. Guck and Carmen Mas Machuca and Wolfgang Kellerer},
  year = {2019},
  eprint = {arXiv:1903.00436},
  primaryClass={cs.NI}
}

@article{litinski2019gameofsurfacecodes,
  doi = {10.22331/q-2019-03-05-128},
  url = {https://doi.org/10.22331/q-2019-03-05-128},
  title = {A {G}ame of {S}urface {C}odes: {L}arge-{S}cale {Q}uantum {C}omputing with {L}attice {S}urgery},
  author = {Litinski, Daniel},
  journal = {{Quantum}},
  issn = {2521-327X},
  publisher = {{Verein zur F{\"{o}}rderung des Open Access Publizierens in den Quantenwissenschaften}},
  volume = {3},
  pages = {128},
  month = mar,
  year = {2019}
}

@article{lin2025spatiallyparalleldecoding,
doi = {10.1088/2058-9565/adc6b6},
url = {https://dx.doi.org/10.1088/2058-9565/adc6b6},
year = {2025},
month = {apr},
publisher = {IOP Publishing},
volume = {10},
number = {3},
pages = {035007},
author = {Fuhui Lin, Sophia and Peterson, Eric C and Sankar, Krishanu and Sivarajah, Prasahnt},
title = {Spatially parallel decoding for multi-qubit lattice surgery},
journal = {Quantum Science and Technology}
}

@article{skoric2023parallelwindowdecoding,
  author = {Luka Skoric and Dan E. Browne and Kenton M. Barnes and Neil I. Gillespie and Earl T. Campbell},
  journal = {Nature Communications},
  number = {7040},
  title = {Parallel window decoding enables scalable fault tolerant quantum computation},
  volume = {14},
  year = {2023}
}

@article{horsman2012latticesurgery,
doi = {10.1088/1367-2630/14/12/123011},
url = {https://dx.doi.org/10.1088/1367-2630/14/12/123011},
year = {2012},
month = {dec},
publisher = {IOP Publishing},
volume = {14},
number = {12},
pages = {123011},
author = {Horsman, Dominic and Fowler, Austin G and Devitt, Simon and Meter, Rodney Van},
title = {Surface code quantum computing by lattice surgery},
journal = {New Journal of Physics}
}

@misc{bombin2023modulardecoding,
  title = {Modular decoding: parallelizable real-time decoding for quantum computers},
  author = {Héctor Bombín and Chris Dawson and Ye-Hua Liu and Naomi Nickerson and Fernando Pastawski and Sam Roberts},
  year = {2023},
  eprint = {arXiv:2303.04846},
  primaryClass={quant-ph}
}

@article{gidney2021stim,
  doi = {10.22331/q-2021-07-06-497},
  url = {https://doi.org/10.22331/q-2021-07-06-497},
  title = {Stim: a fast stabilizer circuit simulator},
  author = {Gidney, Craig},
  journal = {{Quantum}},
  issn = {2521-327X},
  publisher = {{Verein zur F{\"{o}}rderung des Open Access Publizierens
                in den Quantenwissenschaften}},
  volume = 5,
  pages = 497,
  month = jul,
  year = 2021
}

@article{cao2019quantumchemistry,
author = {Cao, Yudong and Romero, Jonathan and Olson, Jonathan P. and Degroote, Matthias and Johnson, Peter D. and Kieferová, Mária and Kivlichan, Ian D. and Menke, Tim and Peropadre, Borja and Sawaya, Nicolas P. D. and Sim, Sukin and Veis, Libor and Aspuru-Guzik, Alán},
title = {Quantum Chemistry in the Age of Quantum Computing},
journal = {Chemical Reviews},
volume = {119},
number = {19},
pages = {10856-10915},
year = {2019},
doi = {10.1021/acs.chemrev.8b00803},
    note ={PMID: 31469277},
URL = { 
        https://doi.org/10.1021/acs.chemrev.8b00803
},
eprint = { 
        https://doi.org/10.1021/acs.chemrev.8b00803
}
}

@article{shor1994algo,
  title = {Algorithms for quantum computation: discrete logarithms and factoring},
  author={Peter W. Shor},
  journal = {Proceedings 35th Annual Symposium on Foundations of Computer Science},
  pages = {124-134},
  publisher = {IEEE},
  year = {1994}
}

@misc{gidney2025howtofactor,
  title = {How to factor 2048 bit RSA integers with less than a million noisy qubits},
  author = {Craig Gidney},
  year = {2025},
  eprint = {arXiv:2505.15917},
  primaryClass={quant-ph}
}

@article{biamonte2017qml,
  author = {Jacob Biamonte and Peter Wittek and Nicola Pancotti and Patrick Rebentrost and Nathan Wiebe and Seth Lloyd},
  journal = {Nature},
  pages = {195--202},
  title = {Quantum machine learning},
  volume = {549},
  year = {2017}
}

@article{dennis2002topologicalquantummemory,
  author = {Eric Dennis and Alexei Kitaev and Andrew Landahl and John Preskill},
  journal = {J. Math. Phys.},
  pages = {4452--4505},
  title = {Topological quantum memory},
  volume = {43},
  issue = {9},
  year = {2002}
}

@misc{bravyi1998quantumcodesboundary,
  title = {Quantum codes on a lattice with boundary},
  author = {S. B. Bravyi and A. Yu. Kitaev},
  year = {1998},
  eprint = {arXiv:quant-ph/9811052},
  primaryClass={quant-ph}
}

@article{kitaev2003faulttolerant,
title = {Fault-tolerant quantum computation by anyons},
journal = {Annals of Physics},
volume = {303},
number = {1},
pages = {2-30},
year = {2003},
issn = {0003-4916},
doi = {https://doi.org/10.1016/S0003-4916(02)00018-0},
url = {https://www.sciencedirect.com/science/article/pii/S0003491602000180},
author = {A.Yu. Kitaev}
}

@article{chan2023actisstrictlylocal,
  doi = {10.22331/q-2023-11-14-1183},
  url = {https://doi.org/10.22331/q-2023-11-14-1183},
  title = {Actis: {A} {S}trictly {L}ocal {U}nion–{F}ind {D}ecoder},
  author = {Chan, Tim and Benjamin, Simon C.},
  journal = {{Quantum}},
  issn = {2521-327X},
  publisher = {{Verein zur F{\"{o}}rderung des Open Access Publizierens in den Quantenwissenschaften}},
  volume = {7},
  pages = {1183},
  month = nov,
  year = {2023}
}

@misc{lee2025efficientpostselection,
  title = {Efficient Post-Selection for General Quantum LDPC Codes},
  author = {Seok-Hyung Lee and Lucas English and Stephen D. Bartlett},
  year = {2025},
  eprint = {arXiv:2510.05795},
  primaryClass={quant-ph}
}

@misc{wolanski2024ambiguityclustering,
  title = {Ambiguity Clustering: an accurate and efficient decoder for qLDPC codes},
  author = {Stasiu Wolanski and Ben Barber},
  year = {2024},
  eprint = {arXiv:2406.14527},
  primaryClass={quant-ph}
}

@misc{ziad2024localclusteringdecoder,
  title = {Local Clustering Decoder: a fast and adaptive hardware decoder for the surface code},
  author = {Abbas B. Ziad and Ankit Zalawadiya and Canberk Topal and Joan Camps and György P. Gehér and Matthew P. Stafford and Mark L. Turner},
  year = {2024},
  eprint = {arXiv:2411.10343},
  primaryClass={quant-ph}
}

@misc{liyanage2025heliosv3,
  title = {Network-Integrated Decoding System for Real-Time Quantum Error Correction with Lattice Surgery},
  author = {Namitha Liyanage and Yue Wu and Emmet Houghton and Lin Zhong},
  year = {2025},
  eprint = {arXiv:2504.11805},
  primaryClass={quant-ph}
}

@article{valentino2025quekuf,
author = {Valentino, Federico and Branchini, Beatrice and Conficconi, Davide and Sciuto, Donatella and Santambrogio, Marco D.},
title = {QUEKUF: An FPGA Union Find Decoder for Quantum Error Correction on the Toric Code},
year = {2025},
issue_date = {September 2025},
publisher = {Association for Computing Machinery},
address = {New York, NY, USA},
volume = {18},
number = {3},
issn = {1936-7406},
url = {https://doi.org/10.1145/3733239},
doi = {10.1145/3733239},
journal = {ACM Trans. Reconfigurable Technol. Syst.},
month = aug,
articleno = {33},
numpages = {26}
}

@INPROCEEDINGS{heer2023noveluf,
  author={Heer, Maximilian Jakob and Sozzo, Emanuele Del and Fujii, Keisuke and Sano, Kentaro},
  booktitle={2023 IEEE International Parallel and Distributed Processing Symposium Workshops (IPDPSW)}, 
  title={Novel Union-Find-based Decoders for Scalable Quantum Error Correction on Systolic Arrays}, 
  year={2023},
  volume={},
  number={},
  pages={524-533},
  doi={10.1109/IPDPSW59300.2023.00092}}

@INPROCEEDINGS{heer2023achievingscalable,
  author={Heer, Maximilian Jakob and Wichmann, Jan-Erik R. and Sano, Kentaro},
  booktitle={2023 IEEE International Conference on Quantum Computing and Engineering (QCE)}, 
  title={Achieving Scalable Quantum Error Correction with Union-Find on Systolic Arrays by Using Multi-Context Processing Elements}, 
  year={2023},
  volume={02},
  number={},
  pages={242-243},
  doi={10.1109/QCE57702.2023.10224}}

@article{barber2025collisionclustering,
    author = {Ben Barber and Kenton M. Barnes and Tomasz Bialas and Okan Buğdaycı and Earl T. Campbell and Neil I. Gillespie and Kauser Johar and Ram Rajan and Adam W. Richardson and Luka Skoric and Canberk Topal and Mark L. Turner and Abbas B. Ziad},
    title = {A real-time, scalable, fast and resource-efficient decoder for a quantum computer},
    journal = {Nature Electronics},
    year = {2025},
    volumes = {8},
    pages = {84--91}
}

@inproceedings{dong2021efficientstepping,
author = {Dong, Xiaojun and Gu, Yan and Sun, Yihan and Zhang, Yunming},
title = {Efficient Stepping Algorithms and Implementations for Parallel Shortest Paths},
year = {2021},
isbn = {9781450380706},
publisher = {Association for Computing Machinery},
address = {New York, NY, USA},
url = {https://doi.org/10.1145/3409964.3461782},
doi = {10.1145/3409964.3461782},
booktitle = {Proceedings of the 33rd ACM Symposium on Parallelism in Algorithms and Architectures},
pages = {184–197},
numpages = {14},
location = {Virtual Event, USA},
series = {SPAA '21}
}

@inproceedings{vedadi2025hybstepping,
author = {Vedadi Gargary, Ashkan and Fuad, Sakib},
title = {Hyb-Stepping: Hybrid Stepping for Parallel Shortest Paths},
year = {2025},
isbn = {9798400714467},
publisher = {Association for Computing Machinery},
address = {New York, NY, USA},
url = {https://doi.org/10.1145/3711708.3723447},
doi = {10.1145/3711708.3723447},
booktitle = {Proceedings of the 1st FastCode Programming Challenge},
pages = {48–54},
numpages = {7},
location = {The Westin Las Vegas Hotel \& Spa, Las Vegas, NV, USA},
series = {FCPC '25}
}

@misc{gidney2024msc,
  title = {Magic state cultivation: growing T states as cheap as CNOT gates},
  author = {Craig Gidney and Noah Shutty and Cody Jones},
  year = {2024},
  eprint = {arXiv:2409.17595},
  primaryClass={quant-ph}
}

@misc{hirano2025efficientmsc,
  title = {Efficient magic state cultivation with lattice surgery},
  author = {Yutaka Hirano and Riki Toshio and Tomohiro Itogawa and Keisuke Fujii},
  year = {2025},
  eprint = {arXiv:2510.24615},
  primaryClass={quant-ph}
}

@misc{akahoshi2025runtimereduction,
  title = {Runtime reduction in lattice surgery utilizing time-like soft information},
  author = {Yutaro Akahoshi and Riki Toshio and Jun Fujisaki and Hirotaka Oshima and Shintaro Sato and Keisuke Fujii},
  year = {2025},
  eprint = {arXiv:2510.21149},
  primaryClass={quant-ph}
}

@article{griffiths2024ufwouf,
  title = {Union-find quantum decoding without union-find},
  author = {Griffiths, Sam J. and Browne, Dan E.},
  journal = {Phys. Rev. Res.},
  volume = {6},
  issue = {1},
  pages = {013154},
  numpages = {9},
  year = {2024},
  month = {Feb},
  publisher = {American Physical Society},
  doi = {10.1103/PhysRevResearch.6.013154},
  url = {https://link.aps.org/doi/10.1103/PhysRevResearch.6.013154}
}

@ARTICLE{delfosse2022ufqldpc,
  author={Delfosse, Nicolas and Londe, Vivien and Beverland, Michael E.},
  journal={IEEE Transactions on Information Theory}, 
  title={Toward a Union-Find Decoder for Quantum LDPC Codes}, 
  year={2022},
  volume={68},
  number={5},
  pages={3187-3199},
  doi={10.1109/TIT.2022.3143452}}

@article{tarjan1975efficiency,
author = {Tarjan, Robert Endre},
title = {Efficiency of a Good But Not Linear Set Union Algorithm},
year = {1975},
issue_date = {April 1975},
publisher = {Association for Computing Machinery},
address = {New York, NY, USA},
volume = {22},
number = {2},
issn = {0004-5411},
url = {https://doi.org/10.1145/321879.321884},
doi = {10.1145/321879.321884},
journal = {J. ACM},
month = apr,
pages = {215–225},
numpages = {11}
}

@misc{dinca2025errormitigation,
  title = {Error mitigation for logical circuits using decoder confidence},
  author = {Maria Dincă and Tim Chan and Simon C. Benjamin},
  year = {2025},
  eprint = {arXiv:2512.15689},
  primaryClass={quant-ph}
}

@misc{zhou2025errormitigation,
  title = {Error Mitigation of Fault-Tolerant Quantum Circuits with Soft Information},
  author = {Zeyuan Zhou and Shaun Pexton and Aleksander Kubica and Yongshan Ding},
  year = {2025},
  eprint = {arXiv:2512.09863},
  primaryClass={quant-ph}
}

@misc{sunami2025entanglementboosting,
  title = {Entanglement boosting: Low-volume logical Bell pair preparation for distributed fault-tolerant quantum computation},
  author = {Shinichi Sunami and Yutaka Hirano and Toshihide Hinokuma and Hayata Yamasaki},
  year = {2025},
  eprint = {arXiv:2511.10729},
  primaryClass={quant-ph}
}

@misc{xie2026simple,
  title = {Simple, Efficient, and Generic Post-Selection Decoding for qLDPC codes},
  author = {Haipeng Xie and Nobuyuki Yoshioka and Kento Tsubouchi and Ying Li},
  year = {2026},
  eprint = {arXiv:2601.17757},
  primaryClass={quant-ph}
}

\end{document}